\documentclass[10pt]{article}

\usepackage{xspace}
\usepackage[normalem]{ulem}
\usepackage{ragged2e}
\usepackage{url}
\usepackage{mathtools}
\usepackage{amssymb}
\usepackage{amsthm}
\usepackage{empheq}
\usepackage{latexsym}
\usepackage{enumitem}
\usepackage{eurosym}
\usepackage{dsfont}
\usepackage{appendix}
\usepackage{color} 
\usepackage[unicode]{hyperref}
\usepackage{frcursive}
\usepackage[utf8]{inputenc}
\usepackage[T1]{fontenc}
\usepackage{geometry}
\usepackage{multirow}
\usepackage{todonotes}
\usepackage{lmodern}
\usepackage{subcaption}
\usepackage{anyfontsize}
\usepackage{stmaryrd}
\usepackage{natbib}
\usepackage{cleveref}
\usepackage[english]{babel}
\usepackage[english=british]{csquotes}
\usepackage{float}
\usepackage{cancel}

\setcitestyle{numbers,open={[},close={]}}

\definecolor{red}{rgb}{0.7,0.15,0.15}
\definecolor{green}{rgb}{0,0.5,0}
\definecolor{blue}{rgb}{0,0,0.7}
\hypersetup{colorlinks, linkcolor={red},citecolor={green}, urlcolor={blue}}
			
\makeatletter \@addtoreset{equation}{section}

\newtheorem{theorem}{Theorem}[section]
\newtheorem{assumption}[theorem]{Assumption}

\newtheorem{proposition}[theorem]{Proposition}

\newtheorem{definition}[theorem]{Definition}
\newtheorem{remark}[theorem]{Remark}

\DeclareUnicodeCharacter{014D}{\=o}
\def \F{\mathbb{F}}

\def \N{\mathbb{N}}

\def \P{\mathbb{P}}

\def \R{\mathbb{R}}

\def\Gc{{\cal G}}

\def\Zc{{\cal Z}}

\def\d{\mathrm{d}}

\title{Governmental incentives for green bonds investment}
\author{Bastien {\sc Baldacci}\footnote{\'Ecole Polytechnique, CMAP, 91128, Palaiseau, France, bastien.baldacci@polytechnique.edu.} \and Dylan {\sc Possamaï}\footnote{ETH Z\"urich, Department of Mathematics, R\"amistrasse 101, 8092 Z\"urich, Switzerland, dylan.possamai@math.ethz.ch.} }

\begin{document}

\maketitle
\begin{abstract}
Motivated by the recent studies on the green bond market, we build a model in which an investor trades on a portfolio of green and conventional bonds, both issued by the same governmental entity. The government provides incentives to the bondholder in order to increase the amount invested in green bonds. These incentives are, optimally, indexed on the prices of the bonds, their quadratic variation and covariation. We show numerically on a set of French governmental bonds that our methodology outperforms the current tax-incentives systems in terms of green investments. Moreover, it is robust to model specification for bond prices and can be applied to a large portfolio of bonds using classical optimisation methods.


\end{abstract}

\section{Introduction}\label{Section Introduction}

Green bonds are fixed income products, issued by governments or companies to finance their debt. The only difference with the so-called conventional bonds is that they finance environmental or climate-related activities. Since its inception in 2007, the green bonds market has expanded rapidly to reach a total amount issued of $\$100$ billion in $2019$. Corporate and finance companies issue more than $70\%$ of the total amount of green bonds, whereas governments issue approximately $9\%$ of this total, see for example the report of the \textcolor{green}{Financial Stability Board} in \cite{board2015global} or the \citeauthor*{oecd2017green} reports in \cite{oecd2017green,oecd2017investing}. The role of financial markets in promoting environmental policies via the green bonds is well documented in \citeauthor{park2018investors} \cite{park2018investors}. The characteristics of a bond to be defined as `green' is given by the Green Bond Principles, which are `voluntary process guidelines that recommend transparency and disclosure, and promote integrity in the development of the Green Bond market by clarifying the approach for issuance of a Green Bond', see the definition in the guidelines \cite{international2016green}, published by the \textcolor{green}{ICMA}. These principles led the green bonds to become a standardised asset class, part of the traditional asset allocation. There is an important literature on the influence of green bonds on gas emissions and environmental ratings. In \citeauthor*{flammer2020green} \cite{flammer2018corporate,flammer2020green}, the author shows that the stock of a company responds positively to the announcement of green bond issues, and these issuances lead to an improvement of the environmental performance. The pricing and ownership of green bonds in the United States is studied in \citeauthor*{baker2018financing} \cite{baker2018financing}, where the authors show in particular that green municipal bonds are issued at a premium to otherwise similar ordinary bonds. Similarly, the impact of corporate green bonds on the credit quality of the issuer and on the shareholders is well documented by \citeauthor{tang2020shareholders} \cite{tang2020shareholders}. In \citeauthor*{de2020environmental} \cite{de2020environmental}, the authors show how green investments can help companies to reduce their greenhouse gas emissions by raising their cost of capital. In particular, they provide empirical evidence on the US markets that an increase of assets managed by green investors lead to a decrease of carbon emission by the companies.   

\medskip
The idea of financing renewable projects through green bonds is even more important since institutional investors, in particular pension funds and asset managers, have been considering the possibility of including sustainable environmental investments in their assets. As such, ``sustainable investing'' now accounts for more than one quarter of total assets under management (AUM) in the United States and more than half in Europe, see the report of the \textcolor{green}{GSIA} \cite{alliance2016global} for a detailed survey on the subject. The motivations of sustainable investing can be the search of higher alpha or lower risk (see \citeauthor*{nilsson2008investment} \cite{nilsson2008investment}, \citeauthor*{bauer2015social} \cite{bauer2015social}, \citeauthor*{kruger2015corporate} \cite{kruger2015corporate}), or the will for a more socially responsible image (see \citeauthor*{hong2009price} \cite{hong2009price}). The two major practices in sustainable investing are exclusionary screening and environmental, social and governance (ESG) integration. Exclusionary screening involves the exclusion of certain assets from the range of eligible investments on ethical grounds, such as the so-called sin stocks, while ESG integration involves under weighting assets with low ESG ratings and over weighting those with high ESG ratings. In \citeauthor*{zerbib2019sustainable} \cite{zerbib2019sustainable}, the author builds a sustainable CAPM based on these two principles and shows how sustainable investing affects asset returns. Although the issuance of green corporate bonds has increased over the last years, the public sector accounts for two-thirds of the investments in sustainable energy infrastructure. This pleads in favour of a greater issuance of green bonds by public entities to finance their sustainable projects, which will be the focus of the present paper.

\medskip
However, there are still several barriers to the development of the green bond market, such as a lack of green bond definition, framework, and transparency. In that regard, \citeauthor*{zerbib2017green} \cite{zerbib2017green,zerbib2019effect} investigates the existence of a yield premium for green bonds. The results show that there exists a small negative premium meaning that the yield of a green bond is lower than that of a conventional bond. In the existing literature, this negative yield differential is mainly attributed to intangible asset creation, which is imperfectly captured in the models of rating agencies, see for example \citeauthor*{porter1995toward} \cite{porter1995toward}, \citeauthor*{ambec2008does} \cite{ambec2008does}, or \citeauthor*{brooks2018effects} \cite{brooks2018effects}. The price difference between a green and a conventional bond is studied in \citeauthor{hachenberg2018green} \cite{hachenberg2018green}, where the authors show that financial and corporate green bonds trade tighter than their conventional counterpart, and governmental bonds on the other hand trade marginally wider. Finally, \citeauthor*{ekeland2019obligation} \cite{ekeland2019obligation} relativize the use of green bonds to finance the ecological transition. As the green bond principles are by no means legally mandatory, and the investors are not necessarily motivated by the green transition, there are no intrinsic difference between a green bond and its conventional counterpart. An important aspect in order to avoid green-washing, that is when the investors use the funding obtained with the green bonds to finance non-sustainable projects, through green bonds is the issuer’s reputation or green third-party verifications, as stated in \citeauthor*{bachelet2019green} \cite{bachelet2019green}. These studies show the several components which slow down the development of the green bonds market. It is therefore important to put in place practical solutions to overcome these constraints. Some mechanisms are already developed by the policy makers to facilitate the investment in this market.

\medskip
Indeed, there are several types of incentives policy-makers can put in place to support green bond issuance, see \citeauthor*{morel2012financing} \cite{morel2012financing}, and \citeauthor*{della2011role} \cite{della2011role}: support for research and development (R\&D), investment incentives (capital grants, loan guarantees and low-interest rate loans), policies which target the cost of investment in capital by hedging or mitigating risk, and tax incentives policies.\footnote{For a complete survey of renewable energy promotion policies, we refer to Table 3 in \citeauthor*{della2011role} \cite{della2011role}.} In particular, tax incentives are attractive from a cost-efficiency perspective, as they can provide a big boost to investment with a relatively low impact on public finances. In \citeauthor{agliardi2019financing} \cite{agliardi2019financing}, the authors show that governmental tax-based incentives play a significant role in scaling up the green bonds market. Finally, tax incentives (accelerated depreciation, tax credits, tax exemptions and rebates) can be provided either to the investor or to the issuer under the following forms.\footnote{The data provided below can be found at \url{https://www.climatebonds.net/policy/policy-areas/tax-incentives}.}

\begin{itemize}
    \item  Tax credit bonds: Bond investors receive tax credits instead of interest payments, so issuers do not pay coupon interests. Instead, they quarterly accrue phantom taxable income and tax credit equal to the amount of phantom income to holders, see \citeauthor{klein2009tax} \cite{klein2009tax}.
    \item Direct subsidy bonds: Bond issuers receive cash rebates from the government to subsidise their net interest payments. This type of incentives is mainly used by US municipalities, see for example \citeauthor*{ang2010build} \cite{ang2010build}. 
    \item Tax-exempt bonds: Bond investors do not have to pay income tax on interest from the green bonds they hold (so issuer can get lower interest rate). This type of tax incentive is typically applied to municipal bonds in the US market, see \citeauthor*{calabrese2016borrowing} \cite{calabrese2016borrowing} for a survey of the use of these tax-incentives. 
\end{itemize}
All these incentives can be modelled as a function of the amount invested in green bonds. However, it should be clear that policy-makers cannot necessarily control or monitor directly the actions of the investor. 
This leads for example to the so-called `green-washing' practice, when the investors use the funding obtained with the green bonds to finance non-sustainable projects, see \citeauthor*{della2011role} \cite{della2011role}. Moreover, the incentives are not dynamic in the sense that they do not depend on the evolution of market conditions (for example the price differences between green and conventional bonds). Thus, the incentives mechanism in the green bonds market is subject to a moral hazard component. In this article, we propose an alternative to tax incentives policy which is based on contract theory, and designed so as to increase the investment in green bonds. Moral hazard, whose related theory has been developed since the early $70$'s, occurs when one person or entity (the Agent), is able to make decisions and/or take actions on behalf of, or that impact, another person or entity: the Principal. The classical continuous-time setting works as follow: the Principal hires an Agent to manage a `risky' project, represented as a controlled stochastic differential equation. In exchange for the effort he puts into his work, the Agent receives a salary from the Principal which takes the form of a  `contract'. The Principal's goal is to offer a contract to the Agent allowing him to maximise its utility as a function of the terminal value of the project. The problem is addressed by solving a Stackelberg game, in two stages:
\begin{enumerate}
    \item[$(i)$] With a fixed contract, solve the problem of the Agent and obtain its optimal effort given a contract proposed by the Principal.
    \item [$(ii)$] Inject into the problem of the Principal the effort of better response of the Agent previously found, and solve the Principal's problem, providing the optimal contract offered to the Agent.
\end{enumerate}
Our goal is to propose a dynamic incentives model based on the prices and returns of green and conventional bonds issued by a government. We build a Principal-Agent model in which an investor (the Agent) runs a portfolio of green and conventional bonds. Without intervention of the government, the Agent has specific investment targets coming from his strategy. The policy-maker (the Principal) proposes incentives to the investor in order to achieve two objectives:
\begin{enumerate}
    \item[$(i)$] Increase the amount invested in green bonds according to a determined target;
    \item[$(ii)$] maximise the value of the portfolio of bonds issued by the government.
\end{enumerate}
We show that without loss of utility for the government, we can consider incentives which take the form of stochastic integrals with respect to the portfolio process, the price of the bonds and their quadratic (co)variation. In order to propose tractable incentives for a possibly high number of bonds, we propose a form of contract that is based only on the dynamics of the portfolio process, the green bonds, an index of conventional bonds, and their respective quadratic (co)variations. In the case of deterministic short-term rates for the green and conventional bonds, both the Agent and the Principal's problems can be solved by maximising deterministic functions with classical root-finding methods. When a one factor stochastic volatility model is used for short-term rates, we have to rely on stochastic control theory and determining the incentives of the policy-maker is equivalent to solve a high-dimensional, Hamilton-Jacobi-Bellman equation. 

\medskip
What we propose in this paper is aimed to be used by governments as an alternative to the existing tax incentives, in order to increase the investment in green bonds. We summarise below the key features of our approach.
\begin{itemize}
    \item The methodology we develop is completely tractable from a numerical point of view, thus the incentives can be designed on a large set of bonds.
    \item The remuneration we propose take into account the moral hazard between the investor and the government: the amount invested in the bonds is observed but not controlled by the government. 
    \item The form of the optimal incentives is robust to model error: we show numerically that a more complex dynamics of the short-term rates of the bonds does not lead to an important loss in utility for the government, and causes minor variations in the form of the incentives.
    \item On a one-year horizon, the incentives show a rather constant behaviour. By using this, we show that the optimal incentives can be directly implemented with tradable financial products such as futures, log-contracts and variance swaps on the bonds.
    \item We compare our methodology with the current tax-incentives policy and show that, on a one-year period for a same target in green investments, our incentives policy leads to a higher value of the portfolio of bonds ($15\%$ to $20\%$ on average).
\end{itemize}
In the numerical experiments, we provide general guidelines for the government to calibrate the model parameters, in particular the risk aversions, according to its objectives.

\medskip
This article makes several contributions to the literature. First, to the best of our knowledge, it offers the first Principal-Agent framework to tackle the design of governmental incentives for green bonds. Contrary to articles like \citeauthor*{zerbib2019effect} \cite{zerbib2019effect} and \citeauthor*{febi2018impact} \cite{febi2018impact}, where the authors provide a thorough descriptive analysis of the green bond market (risk premium, liquidity premium, ...) and examine the impact of green investing, our article focuses on answering a practical incentives problem from a quantitative viewpoint. The comparison with existing incentives policy on a set of French governmental bonds shows the benefits of our method for the government. The article contributes also to the Principal-Agent literature with volatility control, of which we give a brief overview.\footnote{This literature has been growing since the study of the well-posedness of second-order backward stochastic differential equations, see for example \citeauthor*{possamai2018stochastic} \cite{possamai2018stochastic}, or \citeauthor*{soner2012wellposedness} \cite{soner2012wellposedness}. A rigorous study of the Principal-Agent problem with volatility control in a general case can be found in \citeauthor*{cvitanic2018dynamic} \cite{cvitanic2018dynamic}.} Contrary to the papers of  \citeauthor*{sung1995linearity} \cite{sung1995linearity}, \citeauthor*{ou2003optimal} \cite{ou2003optimal}, the Principal observes the whole path of the controlled output process. Moreover, in our framework, moral hazard arises from unobservable sources of risk. In \citeauthor*{lioui2013optimal} \cite{lioui2013optimal}, the authors consider a first-best problem with volatility control and assume that the agent has enough bargaining power to make the contract a linear function of the output and a benchmark risk factor. Another model is the one of \citeauthor*{leung2014continuous} \cite{leung2014continuous} where moral hazard with respect to the volatility arises because of the un-observability of the risk factors by the Principal and an exogenous source of risk multiplying the volatility of the Agent. These works are linked to the problem of ambiguity aversion on volatility and drift of the output process, see among others \citeauthor{chen2018managerial} \cite{chen2018managerial}, \citeauthor*{hernandez2019moral} \cite{hernandez2019moral}, \citeauthor*{mastrolia2015moral} \cite{mastrolia2015moral}, 
\citeauthor{sung2015optimal} \cite{sung2015optimal}. There is also a growing literature on the application of Principal-Agent with volatility control to the electricity market, see for example \citeauthor*{elie2019mean} \cite{elie2019mean}, \citeauthor*{aid2018optimal} \cite{aid2018optimal}. Finally, we emphasise that the modelling framework of this article is directly inspired by the one in \citeauthor*{cvitanic2017moral} \cite{cvitanic2017moral}, where the authors consider the problem of delegated portfolio management and identify a family of admissible contracts for which the optimal agent’s action is explicitly characterised. We extend this framework by allowing stochastic drift of the assets held by the Agent, and adapt it to our context.

\medskip
The paper is organised as follows. In \Cref{sec_framework}, we present our framework and modelling assumptions. In \Cref{sec_deterministic_rates}, we solve the problems of the investor and the government with moral hazard and deterministic short rates. We present the numerical results in \Cref{sec_numerical}. Finally, we write in \Cref{sec_weak_formulation} the weak formulation of the control problem, while in \Cref{sec_stoch_rates} we solve the problem in the case of stochastic short rates.

\medskip
\textbf{Notations:} For $(v_1,v_2) \in \mathbb{R}^d$, $v_1\cdot v_2 \in \mathbb{R}$ denote the scalar product between $v_1$ and $v_2$ whereas $v_1 \circ v_2 \in \mathbb{R}^d$ is the component-wise multiplication of the vectors. Let $\N^\star$ be the set of all positive integers. For any $(\ell,c)\in \mathbb{N}^{\star} \times \mathbb{N}^{\star}$, $\mathcal{M}_{\ell,c}(\mathbb{R})$ will denote the space of $\ell\times c$ matrices with real entries. Elements of the matrix $M\in \mathcal{M}_{\ell,c}$ are denoted $(M_{i,j})_{ (i,j)\in\{1,\dots\ell\}\times\{1,\dots c\}}$ and the transpose of $M$ is denoted $M^{\top}$. We identify $\mathcal{M}_{\ell,1}$ with $\mathbb{R}^{\ell}$. When $\ell=c$, we let $\mathcal{M}_{\ell}(\mathbb{R}):=\mathcal{M}_{\ell,\ell}(\mathbb{R})$. For any $x\in \mathcal{M}_{\ell,c}(\mathbb{R})$, and for any $i\in\{1,\dots\ell\}$ and $j\in\{1,\dots,c\}$, $x_{i,:}\in \mathcal{M}_{1,c}(\mathbb{R})$, and $x_{:,j}\in \mathbb{R}^{\ell}$ denote respectively the $i$-th row and the $j$-th column of $M$. For any $d\in\mathbb{N}^\star$, $\mathbb{S}_d$ is the space of $d\times d$-dimensional symmetric matrices. For any $(\ell,c)\in \mathbb{N}^\star\times\mathbb{N}^\star$, we define $\mathrm{I}_{\ell}$ as the identity matrix of $\mathcal{M}_\ell(\mathbb{R})$, and $\mathbf{0}_{\ell,c}$ as a matrix in $\mathcal{M}_{\ell,c}(\mathbb{R})$ with all entries equal to zero. We define the function $\text{diag}:\mathbb{R}^d \longrightarrow \mathcal{M}_d(\mathbb{R})$ such that for $v\in \mathbb{R}^d$, and any $(i,j)\in\{1,\dots,d\}^2$, $\text{\rm diag}(v)_{i,j}:= v_i$ if $i=j$, and $0$ otherwise. For $x\in\mathcal{M}_{\ell,c}(\mathbb{R})$, we define $\|x\|_2 := \sum_{(i,j)\in\{1,\dots,\ell\}\times \{1,\dots,c\}}x_{i,j}^2$.

\section{Framework}\label{sec_framework}
Throughout the article, we work on a filtered probability space $(\Omega,\mathcal{F},\mathbb{P})$ under which all stochastic processes are defined. We refer to \Cref{sec_weak_formulation} for the rigorous weak formulation of the problem, and we intend the present section to have a more accessible (and therefore more heuristic) flavour.

\medskip
We consider an investor wishing to develop his bonds' portfolio. He wants to acquire both green and conventional bonds issued by the same governmental entity or company with possible different amounts issued and different maturities. We assume that we are given a time horizon $T>0$, and positive integers $d^g$ and $d^c$. The investor manages, over the horizon $[0,T]$,  $d^{g}$ green bonds, $d^{c}$ conventional bonds, and an index of conventional bonds of dynamics given by\footnote{We define the index as an average of the dynamics of the conventional bonds. In practice, the investor may trade a large quantity of conventional bonds and only a couple of green bonds. Thus, we argue that it is more convenient for the government to index the remuneration proposed on an average dynamics of conventional bonds in order to have more granularity for the green bonds' incentives.}
\begin{align}\label{dynamics_bonds}
\begin{split}
    & \d P^{g}(t,T^{g}) := P^{g}(t,T^{g})\circ \Big(\big(r^{g}(t) + \eta^{g}(t)\circ \sigma^{g}(t) \big)\mathrm{d}t + \text{diag}\big(\sigma^{g}(t)\big) \d W^{g}_t\Big), \\
    & \d P^{c}(t,T^{c}) := P^{c}(t,T^{c})\circ\Big(\big(r^{c}(t) + \eta^{c}(t)\circ \sigma^{c}(t) \big)\mathrm{d}t + \text{diag}\big(\sigma^{c}(t)\big) \d W^{c}_t\Big), \\
    & \d I_t := I_t\big(\mu^I(t) \mathrm{d}t + \sigma^I(t) \d W_t^I\big).     
\end{split}
\end{align}
In the above equations, $T^g$ is an $\mathbb{R}^{d^g}$-valued vector representing the maturities of each green bond and $T^c$ is a $\mathbb{R}^{d^c}$-valued vector representing the maturities of each conventional bond. The functions $\mu^I:[0,T]\longrightarrow \mathbb{R}$, $\sigma^I:[0,T]\longrightarrow \mathbb{R}$ represent respectively the drift and volatility of the index of conventional of bonds $(I_t)_{t\in [0,T]}$. Similarly, the functions $r^g:[0,T]\longrightarrow \mathbb{R}^{d^g}$, $r^c:[0,T]\longrightarrow \mathbb{R}^{d^c}$ represent the vectors of short-term rate of the green and conventional bonds, and the functions  $\eta^g:[0,T]\longrightarrow \mathbb{R}^{d^g}$, $\eta^c:[0,T]\longrightarrow \mathbb{R}^{d^c}$ represent the vectors of risk premia of the green and conventional bonds, while functions $\sigma^g:[0,T]\longrightarrow \mathbb{R}^{d^g},$ $\sigma^c:[0,T]\longrightarrow \mathbb{R}^{d^c}$ represent the vector of volatilities of the green and conventional bonds. The processes $(W_t^g)_{t\in[0,T]},(W_t^c)_{t\in[0,T]},(W_t^I)_{t\in[0,T]}$ are respectively $\mathbb{R}^{d^g},\mathbb{R}^{d^c}$ and $\mathbb{R}$-valued Brownian motions. Finally
\begin{align*}
    W :=
    \begin{pmatrix}
    W^g \\
    W^c \\
    W^I
\end{pmatrix}
\end{align*}
is an $\mathbb{R}^{d^g+d^c+1}$-valued Brownian motion, whose co-variance structure is given by $\d\langle W \rangle_t = \Sigma\mathrm{d}t$, where 
\begin{align*}
\Sigma\in \mathcal{M}_{d^g+d^c+1}(\mathbb{R}), \; \Sigma := 
\begin{pmatrix}
\Sigma^g & \Sigma^{g,c} & \Sigma^{g,I} \\
\Sigma^{g,c} & \Sigma^c & \Sigma^{c,I} \\
\Sigma^{g,I} & \Sigma^{c,I} & \Sigma^I 
\end{pmatrix},
\end{align*}
with
\begin{align*}
    & \Sigma^g \in \mathcal{M}_{d^g}(\mathbb{R}),\; \Sigma_{i,j}^g :=  \rho_{i,j}^{g} \in [-1,1] \text{ if } i\neq j, \; 1 \text{ otherwise},\; (i,j)\in \{1,\dots,d^g\}^2,  \\
    & \Sigma^c \in \mathcal{M}_{d^c}(\mathbb{R}),\;  \Sigma^c_{i,j} :=  \rho_{i,j}^{c}\in [-1,1]  \text{ if } i\neq j,\; 1 \text{ otherwise},\; (i,j)\in \{1,\dots,d^g\}^2, \\
    & \Sigma^{g,c} \in \mathcal{M}_{d^g,d^c}(\mathbb{R}), \; \Sigma^{g,c}_{i,j} := \rho_{i,j}^{gc}\in [-1,1] ,\; (i,j)\in \{1,\dots,d^g\}\times \{1,\dots,d^c\},  \\
    & \Sigma^{g,I}\in \mathbb{R}^{d^g},\;  \Sigma^{g,I}_i := \rho_{i}^{gI}\in [-1,1] ,\; i\in \{1,\dots,d^g\},   \\
    & \Sigma^{c,I} \in \mathbb{R}^{d^c},\;  \Sigma^{c,I}_i:= \rho_{i}^{cI}\in [-1,1] ,\; i\in \{1,\dots,d^c\}. 
\end{align*}
\begin{remark}
All these quantities are assumed to be deterministic, in order to derive a governmental incentive that is tractable for a large number of bonds. We will show in {\rm \Cref{sec_stoch_rates}} that, at the expense of a higher computational cost and the use of stochastic control theory, one can also derive incentives for the investor when short-term rates are stochastic. 
In {\rm\Cref{sec_numerical}}, we show numerically that the use of stochastic short-term rates for the green bonds does not impact qualitatively our results. In particular, when the short-term rates are driven by Ornstein-Uhlenbeck processes, the optimal investment policy in this case oscillates slightly around the one obtained with deterministic rates. Thus, the methodology we propose appears to be robust to model specification. 
\end{remark}
Throughout the paper, we use the following technical assumption.
\begin{assumption}\label{assumption_bounded}
The functions $r^g$, $r^c$, $\eta^c$, $\eta^g$, $\sigma^g$, $\sigma^c$, $\mu^I$, and $\sigma^I$ are uniformly bounded on $[0,T]$. 
\end{assumption}
The investment policy is defined by a vector of control processes $\pi=(\pi_t^g,\pi_t^c,\pi_t^I)_{t\in [0,T]}\in \mathcal{A}$, representing the amount of money invested at time $t$, where
\begin{align*}
    \mathcal{A} := \Big\{(\pi_t)_{t\in [0,T]}:  K \text{-valued and $\mathbb{F}$-predictable processes} \Big\}.
\end{align*}
is the set of admissible control process, where $K:=[\varepsilon,b_\infty]^{d^g}\times [\varepsilon,b_\infty]^{d^c}\times [\varepsilon,b_\infty],$ for some $0<\varepsilon< b_\infty$\footnote{We force the control processes to be strictly positive so that the density of the canonical process in \Cref{sec_weak_formulation} is invertible and we can define properly the weak formulation of the control problem. Practically, this simply means that the investor ahs to invest in the index, and in at least one of the conventional and one of the green bonds.} and $\mathbb{F}:=(\mathcal{F}_t)_{t\in[0,T]}$ is the natural filtration of the process $(X,W)$ with $X$ defined below. We define the dynamics of the vectors of returns on the bonds as 
\begin{align*}
    & \d R^g(t,T^g) =\big(r^{g}(t) + \eta^{g}(t)\circ \sigma^{g}(t) \big)\mathrm{d}t + \text{diag}\big(\sigma^{g}(t)\big) \d W^{g}_t , \\
    & \d R^c(t,T^c) =\big(r^{c}(t) + \eta^{c}(t)\circ \sigma^{c}(t) \big)\mathrm{d}t + \text{diag}\big(\sigma^{c}(t)\big) \d W^{c}_t ,  \\
    & \d R^I_t =\mu^I(t)\mathrm{d}t + \sigma^I(t) \d W^{I}_t. 
\end{align*}
For every $\pi\in \mathcal{A}$, one can define a probability measure $\mathbb{P}^\pi$\footnote{See \Cref{sec_weak_formulation} for the weak formulation of the control problem, which explains how to construct $\P^\pi$.} such that the dynamics of the value of portfolio of bonds is given by
\begin{align*}
    \d X_t & := \pi_t^g \cdot \d R^g(t,T^g) + \pi_t^c \cdot \d R^c(t,T^c) + \pi_t^I \d R^I_t .
\end{align*}
We also denote by $\mathbb{E}_t^\pi$ the conditional expectation under the probability measure $\mathbb{P}^\pi$ with respect to $\mathcal{F}_t$ for all $t\in[0,T]$. Throughout the investment period $[0,T]$, the investor wants to maintain his investment in bonds at some pre-defined levels, which can be seen as his investment profile. We introduce the vectors $\alpha=(\alpha^g,\alpha^c,\alpha^I)\in \mathbb{R}^{d^g}\times\mathbb{R}^{d^c}\times\mathbb{R}$ and the cost function $k:\mathbb{R}^{d^g}\times\mathbb{R}^{d^c}\times\mathbb{R}\longrightarrow \R$, where for any $p:=(p^g,p^c,p^I)\in \mathbb{R}^{d^g}\times\mathbb{R}^{d^c}\times\mathbb{R}$
\begin{align*}
    k(p) := \frac{1}{2}\beta^g \cdot (p^g - \alpha^g)^2 + \frac{1}{2}\beta^c \cdot (p^c - \alpha^c)^2 + \frac{1}{2}\beta^I (p^I - \alpha^I)^2,
\end{align*}
where $\beta:=(\beta^g,\beta^c,\beta^I)\in \mathbb{R}^{d^g}\times\mathbb{R}^{d^c}\times\mathbb{R}$ are what we coin intensity vectors. For instance, at some time $t\in[0,T]$, the investor pays a cost to move the amount $(\pi_t^{g})_i$ invested in the $i$-th green bond away from the initial target $\alpha^{g}_i$, and this cost is equal to $\frac{1}{2}\beta_i^{g}\big((\pi_t^{g})_i-\alpha_i^{g}\big)^2$. Thus, $(\beta^g,\beta^c,\beta^I)$ represent the cost intensity of changing the investments of the agent: the higher these coefficients, the more incentives the investor will demand to change his investment profile. 

\medskip
In order to modify an investment policy $\pi\in\mathcal{A}$, the government proposes a remuneration to the investor. It takes the form of an  $\mathcal{F}_T$-measurable random variable denoted by $\xi$, and we will see later that the form of remuneration considered is an indexation on the value of the portfolio of bonds as well as the sources of risk of each bond. The optimisation problem of the investor with CARA utility function writes, for a given contract provided by the government, as
\begin{align*}
    V^A(\xi):= \sup_{\pi\in\mathcal{A}}\mathbb{E}^\pi\bigg[U_A\bigg( \xi - \int_0^T k(\pi_s)\d s\bigg)\bigg], \; U_A(x):=-\exp(-\gamma x),
\end{align*}
where $\gamma>0$ is his risk aversion parameter. To ensure that the control problem of the investor is non-degenerate, we impose the following integrability condition on the contracts
\begin{align}\label{non-degeneracy-investor}
    \sup_{\pi\in\mathcal{A}}\mathbb{E}^\pi \Big[\exp(-\gamma' \xi) \Big] <+\infty, \; \text{for some }\gamma' >\gamma. 
\end{align}

\begin{remark}
We emphasise here that the notion of price for a bond is meaningless as it is not quoted on the National Best Bid and Offer $(${\rm NBBO}$)${\rm:} This is an {\rm OTC} market where the liquidity is provided by one or several dealers. In particular, even though there is a quantity defined as the bond price on Bloomberg, it serves only as an indication as the dealers have no obligation to buy or sell at this price. However, especially in the case of treasury bonds, Futures on the bonds are listed on the Chicago Board Of Trade where the notion of price is meaningful. Thus, throughout the article, the notion of bond price must be thought as the price of a future on the considered bond.   
\end{remark}
On the other hand, the government wishes to maximise the portfolio value of the bonds issued while increasing the amount invested in green bonds. Thus, he wants to maximise, on average, the quantity 
\begin{align*}
    X_T -  \sum_{i=1}^{d^g}\int_0^T \kappa\Big(G_i-\big(\hat{\pi}_t^{g}(\xi)\big)_i\Big)^2 \mathrm{d}t,
\end{align*}
where for $i\in\{1,\dots,d^g\}$, $G_i$ is the investment target in the $i$-th green bond of the government entity, $\kappa>0$ is the cost of moving away from the targets $(G_1,\dots,G_{d^g})$ and $\hat \pi(\xi)$ is a best response of the investor to a given contract $\xi$.\footnote{We will see later that there might be several best responses of the Agent. Thus, following the tradition in the moral hazard literature, we assume that the Principal has enough bargaining power to be able to choose the best response of the Agent that maximises his own utility.} We assume that the cost of moving away from the targets is the same for each green bond, meaning that the government does not have different preferences for each bond (this assumption can of course be relaxed). The government also subtracts from this quantity the contract $\xi$ offered to the investor. Thus, his optimisation problem with CARA utility function writes
\begin{align}\label{pb_principal}
    V_0^P = \sup_{\xi\in \mathcal{C}
    }\sup_{\hat\pi \in \mathcal{A}(\xi)}\mathbb{E}^{\hat \pi}\bigg[U_P\bigg(X_T -  \sum_{i=1}^{d^g}\int_0^T \kappa\Big(G_i-\big(\hat \pi_t^{g}(\xi)\big)_i\Big)^2 \mathrm{d}t - \xi\bigg)\bigg],\; U_P(x)=-\exp(-\nu x),
\end{align}
where $\nu>0$ is the risk aversion parameter of the Principal, 
\begin{align*}
    \mathcal{A}(\xi):= \bigg\{\hat\pi\in\mathcal{A}:  V^A(\xi)= \mathbb{E}^{\hat\pi}\bigg[-\exp\bigg(-\gamma \bigg( \xi - \int_0^T k(\hat\pi_s)\d s\bigg)\bigg)\bigg] \bigg\},
\end{align*}
is the set of best-responses of the Agent to a given contract $\xi$ and 
\begin{align*}
    \mathcal{C}=\big\{\xi: \mathbb{R}\text{-valued, }\mathcal{F}_T\text{-measurable random variable such that } V^A(\xi) \geq R, \; \text{and} \; \eqref{non-degeneracy-investor} \text{ is satisfied}\big\},
\end{align*}
is the set of admissible contracts for the government, where $R<0$ is the reservation utility of the investor:  He will not accept to work for Principal (and accept the contract $\xi$) unless the contract is such that his expected utility is above $R$.

\begin{remark}
We consider here that the reservation utility corresponds to the utility function of the investor in the case $\xi=0$, that is
\begin{align*}
    R = V^A(0) = \sup_{\pi\in\mathcal{A}}\mathbb{E}^\pi\bigg[-\exp\bigg(\gamma \int_0^T k(\pi_s)\d s\bigg)\bigg] = -1,
\end{align*}
where the supremum is reached by choosing $\pi=(\alpha^g,\alpha^c,\alpha^I)$. We will see in the following section that the optimal contract proposed by the government will always saturate this constraint, that is the Principal will provide the Agent with the minimum reservation utility $R$ he requires.
\end{remark}

\section{Solving the optimisation problem}\label{sec_deterministic_rates}

\subsection{The optimal contract}

In this section, we derive the optimal governmental incentives proposed to the investor. As it would be unrealistic (and hardly tractable) to offer a compensation based on the whole universe of governmental bonds, we suggest a remuneration based on the green bonds, the value of the portfolio and an index of conventional bonds. This way, the contract is only indexed on $d^g+2$ variables. The optimal incentives are obtained by maximising a deterministic function, which makes the problem easily tractable for a large number of green bonds. We begin this section with the definition of contractible and non-contractible variables.  
\begin{definition}
The set of contractible variables is defined as the $\mathbb{R}^{d^g+2}$-valued process 
\begin{align*}
    B^{\text{\rm obs}} :=
    \begin{pmatrix}
     X \\
     W^g \\
     W^I
    \end{pmatrix}.
\end{align*}
The set of non-contractible variables is defined as the $\mathbb{R}^{d^c}$-valued process $B^{\text{\sout{\rm obs}}}:=W^c,$ with the following dynamics 
\begin{align*}
   & \d B_t^{\text{\rm obs}} := \mu^{\text{\rm obs}}(t,\pi_t) \mathrm{d}t + \Sigma^{\text{\rm obs}}(t,\pi_t) \d W_t, \; \d B_t^{\text{\sout{\rm obs}}} := \mu^{\text{\sout{\rm obs}}} \mathrm{d}t+ \Sigma^{\text{\sout{\rm obs}}} \d W_t,
\end{align*}
where $\mu^{\text{\sout{\rm obs}}}:= 
    \begin{pmatrix}
    \mathbf{0}_{d^c,1}
    \end{pmatrix},$ $\Sigma^{\text{\sout{\rm obs}}}:= 
    \begin{pmatrix}
        \mathbf{0}_{d^c,d^g} & \mathrm{I}_{d^c} & \mathbf{0}_{d^c,1}
    \end{pmatrix}$, 
 and the maps $\mu^{\text{\rm obs}}:[0,T]\times \mathbb{R}^{d^g}\times\mathbb{R}^{d^c}\times\mathbb{R}\longrightarrow \mathbb{R}^{d^g+2}$, as well as $\Sigma^{\text{\rm obs}}:[0,T]\times \mathbb{R}^{d^g}\times\mathbb{R}^{d^c}\times\mathbb{R}\longrightarrow \mathcal{M}_{d^g+2,d^g+d^c+1}(\mathbb{R})$ are defined for any $p:=(p^g,p^c,p^I)\in \mathbb{R}^{d^g}\times\mathbb{R}^{d^c}\times\mathbb{R}$ and $t\in[0,T]$ by
\begin{align*}
    & \mu^{\text{\rm obs}}(t,p) := 
    \begin{pmatrix}
        p^g \cdot \big(r^g(t)+\eta^g(t)\circ\sigma^g(t)\big) + p^c \cdot \big(r^c(t)+\eta^c(t)\circ\sigma^c(t)\big) + p^I \mu^I(t) \\
        \mathbf{0}_{d^g,1} \\
        0
    \end{pmatrix},\\ 
    & \Sigma^{\text{\rm obs}}(t,p) :=
    \begin{pmatrix}
        (p^g \circ \sigma(t)^g)^{\top} &(p^c\circ \sigma(t)^c)^{\top} & p^I \sigma^I(t) \\
        I_{d^g} & \mathbf{0}_{d^g,d^c} & \mathbf{0}_{d^g,1} \\
        \mathbf{0}_{1,d^g} & \mathbf{0}_{1,d^c} & 1 
    \end{pmatrix}.
    \end{align*}
\end{definition}
Finding the optimal contract $\xi$ in the optimisation problem \eqref{pb_principal} is an arduous task, as we search a solution in the space of $\mathcal{F}_T$-measurable random variables. However, see \citeauthor*{cvitanic2018dynamic} \cite{cvitanic2018dynamic}, it has been shown that without reducing the utility of the Principal, we can restrict our study to admissible contracts which have a specific form. In order to describe this result, we need first to introduce additional notations.

\medskip
We define the quantities 
\begin{align*}
    B:=
    \begin{pmatrix}
    B^{\text{\rm obs}} \\
    B^{\text{\sout{\rm obs}}}
    \end{pmatrix}
    , \; \mu(t,p) := \begin{pmatrix}
    \mu^{\text{\rm obs}}(t,p) \\
    \mu^{\text{\sout{\rm obs}}}
    \end{pmatrix}, \; 
     \Sigma(t,p) := \begin{pmatrix}
    \Sigma^{\text{\rm obs}}(t,p) \\
    \Sigma^{\text{\sout{\rm obs}}}
    \end{pmatrix},\; (t,p)\in [0,T]\times \mathbb{R}^{d^g}\times\mathbb{R}^{d^c}\times\mathbb{R}.
\end{align*}
We also will need to introduce the map $h:[0,T]\times \mathbb{R}^{d^g+d^c+2}\times \mathbb{S}_{d^g+d^c+2}(\mathbb{R})\times K\longrightarrow \mathbb{R}$, with
\begin{align*}
    h(t,z,g,p) = -k(p) + z \cdot \mu(t,p) +  \frac{1}{2}\text{\rm Tr}\big[ g\Sigma(t,p)\Sigma (\Sigma(t,p)^{\top}], \; (t,z,g,p)\in [0,T]\times \mathbb{R}^{d^g+d^c+2}\times \mathbb{S}_{d^g+d^c+2}(\mathbb{R})\times K.
\end{align*}
and for all $(t,z,g)\in [0,T]\times \mathbb{R}^{d^g+d^c+2}\times\mathbb{S}_{d^g+d^c+2}(\mathbb{R})$,
\begin{align*}
    \mathcal{O}(t,z,g):=\Big\{\hat p\in K:  \hat p \in \underset{p\in K}{\text{argmax}}\big\{h(t,z,g,p)\big\}\Big\}.
\end{align*}
is the set of the maximisers of $h$ with respect to its last variable, for $(t,z,g)$ given. Following \citeauthor{schal1974selection} \cite{schal1974selection}, there exists at least one Borel-measurable map $\hat \pi:[0,T]\times \mathbb{R}^{d^g+d^c+2}\times\mathbb{S}_{d^g+d^c+2}(\mathbb{R})\longrightarrow K$ such that for every $(t,z,g)\in [0,T]\times \mathbb{R}^{d^g+d^c+2}\times\mathbb{S}_{d^g+d^c+2}(\mathbb{R})$, $\hat \pi(t,z,g)\in \mathcal{O}(t,z,g)$. We denote by $\mathcal{O}$ the corresponding set of all such maps. 
\begin{theorem}\label{thm_1_admissible_contract}
Without reducing the utility of the Principal, we can restrict the study of admissible contracts to the set $\mathcal{C}_1$  where any $\xi\in\mathcal{C}_1\subset \mathcal{C}$ is of the form $\xi=Y_T^{y_0,Z,\Gamma,\hat \pi}$ where for $t\in[0,T]$,
\begin{align}\label{opt_contract_1}
    Y_t^{y_0,Z,\Gamma,\hat \pi} := y_0 + \int_0^t Z_s\cdot \d B_s + \frac{1}{2}\int_0^t\text{\rm Tr}\big[(\Gamma_s + \gamma Z_s Z_s^{\top})d\langle B\rangle_s\big] - \int_0^t h\big(s,Z_s,\Gamma_s,\hat \pi(s,Z_s,\Gamma_s)\big)\mathrm{d}s,
\end{align}
where $y_0\in \mathbb{R}$, $\hat\pi\in\mathcal{O}$ and $(Z,\Gamma)$ are respectively $\mathbb{R}^{d^g+d^c+2}$- and $\mathbb{S}_{d^g+d^c+2}(\mathbb{R})$-valued, $\mathbb{F}$-predictable processes such that {\rm Condition} \eqref{non-degeneracy-investor} is satisfied for $Y_T^{y_0,Z,\Gamma,\hat \pi}$, and $V^A(Y_T^{y_0,Z,\Gamma,\hat \pi})\geq U_A(y_0)$. We denote by $\mathcal{ZG}$ the set of such processes, which is properly defined in {\rm\Cref{zg_admissible}}. Moreover, we have
\begin{align*}
    V^A\big(Y_T^{y_0,Z,\Gamma,\hat \pi}\big) = U_A(y_0),\;    \mathcal{A}\big(Y_T^{y_0,Z,\Gamma,\hat \pi}\big)= \Big\{\big(\hat\pi(t,Z_t,\Gamma_t)\big)_{t\in[0,T]}: \hat{\pi}\in\mathcal{O}, (Z,\Gamma)\in\mathcal{ZG}\Big\}.
\end{align*}
\end{theorem}
The form of the admissible contracts we study deserves some remarks. The term $Z_t\cdot \d B_t$ is a remuneration indexed linearly on the state variables. Contrary to the classical Principal-Agent problem where the agent controls the drift of the output process, see \citeauthor*{sannikov2008continuous} \cite{sannikov2008continuous} for example, the admissible contracts \eqref{opt_contract_1} are not only linear functions of the state variables but depend also linearly on their quadratic variation and covariation. This comes from the fact that by investing in the bonds, the investor controls directly the volatility of the portfolio process $X$. Using standard tools of static hedging, this contract can be replicated using futures, log-contracts and volatility products such as variance swaps, see \Cref{sec_practical_implementation} for details. In particular, this ensures that the contracts we recommend are practically implementable.

\medskip
As stated at the beginning of this section, we wish to build an optimal contract based only on the green bonds, the portfolio process and the index of conventional bonds. In this regard, the form we obtained in \Cref{opt_contract_1} is too general, which is why we are now going to restrict our attention to a slightly smaller class of contracts. We thus define for any $(Z,\Gamma)\in\Zc\Gc$
\begin{align*}
    Z_t=: 
    \begin{pmatrix}
    Z^{\text{\rm obs}}\\
    Z^{\text{\sout{\rm obs}}}
    \end{pmatrix}, \;
    \Gamma = :
    \begin{pmatrix}
    \Gamma^{\text{\rm obs}}& \Gamma^{\text{\rm obs},\text{\sout{\rm obs}}}\\
    \Gamma^{\text{\rm obs},\text{\sout{\rm obs}}} & \Gamma^{\text{\sout{\rm obs}}}
    \end{pmatrix},
\end{align*}
where for Lebesgue-almost every $t\in[0,T]$
\begin{align*}
    & Z_t^{\text{\rm obs}}\in \mathbb{R}^{d^g+2},\; Z_t^{\text{\sout{\rm obs}}}\in  \mathbb{R}^{d^c},\; \Gamma_t^{\text{\rm obs}} \in \mathbb{S}_{d^g+2}(\mathbb{R}),\; \Gamma_t^{\text{\sout{\rm obs}}} \in \mathbb{S}_{d^c}(\mathbb{R}), \;\Gamma_t^{\text{\rm obs},\text{\sout{\rm obs}}} \in \mathcal{M}_{d^g+2,d^c}(\mathbb{R}).
\end{align*}
We then consider a simplified Hamiltonian $h^{\text{obs}}:[0,T]\times \mathbb{R}^{d^g+2}\times \mathbb{S}_{d^g+2}(\mathbb{R})\times K\longrightarrow \mathbb{R}$ given by 
\begin{align*}
    h^{\text{obs}}(t,z^{\text{obs}},g^{\text{obs}},p) = -k(p) + z^{\text{obs}} \cdot \mu^{\text{obs}}(t,p) +  \frac{1}{2}\mathrm{Tr}\big[ g^{\text{obs}}\Sigma^{\text{obs}}(t,p)\Sigma (\Sigma^{\text{obs}}(t,p)^{\top})\big], 
\end{align*}
and for all $(t,z^{\text{obs}},g^{\text{obs}})\in [0,T]\times \mathbb{R}^{d^g+2}\times\mathbb{S}_{d^g+2}(\mathbb{R})$, we define
\begin{align*}
    \mathcal{O}^{\text{obs}}(t,z^{\text{obs}},g^{\text{obs}}):=\Big\{\hat p\in K: \hat p \in \underset{p\in K}{\text{argmax}}\big\{h^{\text{obs}}(t,z^{\text{obs}},g^{\text{obs}},p)\big\}\Big\}.
\end{align*}
Following again \citeauthor*{schal1974selection} \cite{schal1974selection}, there exists at least one Borel-measurable map $\hat \pi:[0,T]\times \mathbb{R}^{d^g+2}\times\mathbb{S}_{d^g+2}(\mathbb{R})\longrightarrow K$ such that for every $(t,z^{\text{obs}},g^{\text{obs}})\in [0,T]\times \mathbb{R}^{d^g+2}\times\mathbb{S}_{d^g+2}(\mathbb{R})$, $\hat \pi(t,z^{\text{obs}},g^{\text{obs}})\in \mathcal{O}^{\text{obs}}(t,z^{\text{obs}},g^{\text{obs}})$, and we let $\mathcal{O}^{\text{obs}}$ be the corresponding set of all such maps.

\medskip
We can now state precisely the class of contracts we are concerned with in this paper.
\begin{assumption}
We consider the subset of contracts
\begin{align*}
\mathcal{C}_2:=\Big\{Y_T^{y_0,Z,\Gamma,\hat \pi} \in \mathcal{C}_1:Z^{\text{\sout{\rm obs}}}=\mathbf{0}_{d^c},\Gamma^{\text{\sout{\rm obs}}}=\mathbf{0}_{d^c,d^c},\Gamma^{\text{\rm obs},\text{\sout{\rm obs}}}=\mathbf{0}_{d^g+2,d^c} \Big\}.
\end{align*}
In particular, any $\xi\in \mathcal{C}_2$ is of the form $\xi=Y_T^{y_0,Z^{\text{\rm obs}},\Gamma^{\text{\rm obs}},\hat \pi}$, where for any $t\in[0,T]$,
\begin{align}\label{opt_contract}
\begin{split}
    Y_t^{y_0,Z^{\text{\rm obs}},\Gamma^{\text{\rm obs}},\hat \pi} :=  y_0 + \int_0^t & Z_s^{\text{\rm obs}}\cdot \d B_s^{\text{\rm obs}} + \frac{1}{2}\text{\rm Tr}\Big[\big(\Gamma^{\text{\rm obs}}_s + \gamma Z_s^{\text{\rm obs}}(Z_s^{\text{\rm obs}})^{\top}\big)d\langle B^{\text{\rm obs}}\rangle_s\Big] \\
    & - h^{\text{\rm obs}}\Big(s,Z_s^{\text{\rm obs}},\Gamma_s^{\text{\rm obs}},\hat \pi(s,Z_s^{\text{\rm obs}},\Gamma_s^{\text{\rm obs}}))\Big)\mathrm{d}s,    
\end{split}
\end{align}
where $y_0\geq 0$, $\hat\pi\in\mathcal{O}^{\text{\rm obs}}$ and $(Z^{\text{\rm obs}},\Gamma^{\text{\rm obs}})\in\mathcal{ZG}^{\text{\rm obs}}$ with 
\begin{align*}
    \mathcal{ZG}^{\text{\rm obs}}:= \Big\{&(Z^{\text{\rm obs}},\Gamma^{\text{\rm obs}}): \mathbb{R}^{d^g+2}\times\mathbb{S}_{d^g+2}(\mathbb{R})\text{\rm-valued, }\mathbb{F}\text{\rm-predictable, s.t. } Y_T^{y_0,Z^{\text{\rm obs}},\Gamma^{\text{\rm obs}},\hat \pi} \in\mathcal{C}_2\Big\}.
\end{align*}
\end{assumption}
The optimisation problem of the government that we now consider is\footnote{We use the notation $\mathbb{E}^{(\hat \pi(t,Z_t,\Gamma_t))_{t\in[0,T]}}[\cdot]=:\mathbb{E}^{\hat \pi(Z,\Gamma)}[\cdot]$ }
\begin{align}\label{pb_principal_2}
\begin{split}
    \widetilde{V}_0^P = \sup_{y_0\geq 0}\sup_{(Z^{\text{obs}},\Gamma^{\text{obs}},\hat \pi)\in \mathcal{ZG}^{\text{obs}}\times \mathcal{O}^{\text{obs}} }\mathbb{E}^{\hat \pi(Z,\Gamma)}\bigg[U_P\bigg( X_T -  \sum_{i=1}^{d^g}\int_0^T \kappa\Big(G_i-\big(\hat \pi^{g}(t,Z_t^{\text{obs}},\Gamma_t^{\text{obs}})\big)_i\Big)^2 \mathrm{d}t  -Y_T^{y_0,Z^{\text{obs}},\Gamma^{\text{obs}},\hat \pi}\bigg)\bigg],    
\end{split}
\end{align}
This assumption allows us to consider more tractable contracts for a large portfolio of bonds, even if we consider less general contracts compared to \eqref{opt_contract_1}. Moreover, as the objective of the government is to encourage the acquisition of green bonds, it is natural to consider a more granular contract with respect to the green bonds and to use only the index of conventional bonds as a representative contractible variable of this set of bonds. As we used only deterministic functions to model the risk premium, short-term rate and volatility processes, the optimal incentives of the government can be obtained by maximising a deterministic function, which leads to the following theorem.

\begin{theorem}[Main result] The optimal contract $\xi^\star\in \mathcal{C}_2$ is given by
\begin{align}\label{opt_contract_fin}
\begin{split}
   \xi^\star = Y_T^{0,z^{\star\text{obs}},g^{\star\text{obs}},\pi^\star} = \int_0^T & z^{\star\text{\rm obs}}(t)\cdot \d B_t^{\text{\rm obs}} + \frac{1}{2}\mathrm{Tr}\Big[\big(g^{\star\text{\rm obs}}(t) + \gamma z^{\star\text{\rm obs}}(t)(z^{\star\text{\rm obs}}(t))^{\top}\big)d\langle B^{\text{\rm obs}}\rangle_t\Big] \\
   & - h^{\text{\rm obs}}\Big(t,z^{\star\text{\rm obs}}(t),g^{\star\text{\rm obs}}(t),\pi^\star\big(t,z^{\star\text{\rm obs}}(t),g^{\star\text{\rm obs}}(t)\big)\Big)\mathrm{d}t,    
\end{split}
\end{align}
where for all $t\in[0,T]$,  $z^{\star\text{\rm obs}}(\cdot)$, $g^{\star\text{\rm obs}}(\cdot)$, $\pi^\star\big(\cdot,z^{\star\text{\rm obs}}(\cdot),g^{\star\text{\rm obs}}(\cdot)\big)$ are deterministic functions of time, solving
\begin{align}\label{opt_pb_principal}
\begin{split}
    &\sup_{(z,g,\hat \pi)\in P\times\mathcal{O}^{\text{obs}} } \mathcal{H}\big(t,z,g,\hat \pi(t,z,g)\big),
\end{split}
\end{align}
where $P:=\mathbb{R}^{d^g+2}\times\mathbb{S}_{d^g+2}(\mathbb{R})$ and $\mathcal{H}:[0,T]\times P \times K \longrightarrow \mathbb{R}$ is given by
\begin{align*}
\begin{split}
    \mathcal{H}(t,z,g,p) := & -\sum_{i=1}^{d^g} \big(G_i -p_i\big)^2 - \frac{1}{2}\mathrm{Tr}\Big[(g + \gamma zz^{\top})\Sigma^{\text{\rm obs}}\big(t,p\big)\Sigma(\Sigma^{\text{\rm obs}}\big(t,p\big)^{\top}\big)\Big] \\
    & + h^{\text{\rm obs}}\big(t,z,g,p\big) + \Big(\mu^{\text{\rm obs}}\big(t,p\big)\Big)_1- z^{\top}\mu^{\text{\rm obs}}\big(t,p\big) \\
    &  -\frac{1}{2}\nu^2 \Big(\Big(\Sigma^{\text{\rm obs}}\big(t,p\big)\Big)_{1,:}-z^{\top}\Sigma^{\text{\rm obs}}\big(t,p\big)\Big)^{\top} \Sigma \Big(\Big(\Sigma^{\text{\rm obs}}\big(t,p\big)\Big)_{1,:}-z^{\top}\Sigma^{\text{\rm obs}}\big(t,p\big)\Big). 
\end{split}
\end{align*}
Moreover
\begin{align*}
    \widetilde{V}_0^P = U_P\bigg(\int_0^T  \mathcal{H}\Big(t,z^{\star,\text{obs}}(t),g^{\star,\text{obs}}(t),\pi^\star\big(t,z^{\star,\text{obs}}(t),g^{\star,\text{obs}}(t)\big)\Big)\mathrm{d}t \bigg).
\end{align*}
\end{theorem}
\begin{proof}
The term in the exponential of the optimisation problem \eqref{pb_principal_2} is a linear function of $y_0$ hence the reservation utility of the investor is saturated using $y_0^\star =0$. Define for any martingale $M$ the operator
\begin{align*}
    \mathcal{E}(M)_T := \exp\bigg(-\nu M_T + \frac{1}{2}\nu^2 \langle M \rangle_T \bigg). 
\end{align*}
The government has now to solve
\begin{align*}
    \sup_{(Z,\Gamma,\hat \pi)\in \mathcal{ZG}^{\text{obs}}\times \mathcal{O}^{\text{obs}}
    }\mathbb{E}^{\hat\pi(Z,\Gamma)}\Bigg[&U_P\bigg(\int_0^T  \bigg(\big(\mu^{\text{\rm obs}}\big(t,\hat \pi(t,Z_t,\Gamma_t)\big)\Big)_1 - \sum_{i=1}^{d^g}\big(G_i -\hat \pi_i(t,Z_t,\Gamma_t)\big)^2 \\
    & - \frac{1}{2}\mathrm{Tr}\Big[\Big(\Gamma\big(t,\hat\pi(t,Z_t,\Gamma_t)\big) + \gamma Z_tZ_t^{\top}\Big)\Sigma^{\text{\rm obs}}\big(t,\hat\pi(t,Z_t,\Gamma_t)\big)\Sigma \Big(\Sigma^{\text{\rm obs}}\big(t,\hat\pi(t,Z_t,\Gamma_t)\big)\Big)^{\top}\Big] \\
    & + h^{\text{\rm obs}}\big(t,Z_t,\Gamma_t,\hat\pi(t,Z_t,\Gamma_t)\big)\bigg)\mathrm{d}t \bigg) \\
    &\times \exp\bigg(-\nu\int_0^T\Big(\Big(\Sigma^{\text{\rm obs}}\big(t,\hat\pi(t,Z_t,\Gamma_t)\big)\Big)_{0,:} - Z_t^{\top}\Sigma^{\text{\rm obs}}\big(t,\hat\pi(t,Z_t,\Gamma_t)\big)  \Big) \d W_t \bigg)\Bigg]. 
\end{align*}
We make appear the stochastic exponential so that the previous supremum becomes
\begin{align*}
    \sup_{(Z,\Gamma,\hat \pi)\in \mathcal{ZG}^{\text{obs}}\times \mathcal{O}^{\text{obs}}
    }\mathbb{E}^{\hat\pi(Z,\Gamma)}\Bigg[&U_P\bigg(\int_0^T  \mathcal{H}\big(t,Z_t,\Gamma_t,\hat\pi(t,Z_t,\Gamma_t)\big)\mathrm{d}t \bigg)\\
    &\times \mathcal{E}\bigg(\int_0^\cdot\Big(\Big(\Sigma^{\text{\rm obs}}\big(t,\hat\pi(t,Z_t,\Gamma_t)\big)\Big)_{0,:} - Z_t^{\top}\Sigma^{\text{\rm obs}}\big(t,\hat\pi(t,Z_t,\Gamma_t)\big)  \Big) \d W_t\bigg)_T\Bigg].
\end{align*}
As the function $U_P(x)$ is increasing and the expectation of a stochastic exponential is bounded by one, we obtain
\begin{align*}
    \widetilde{V}_0^P \leq U_P\bigg(\int_0^T  \sup_{(z,g,\hat\pi)\in P\times \mathcal{O}^{\text{obs}}}\mathcal{H}\big(t,z,g,\hat\pi(t,z,g)\big)\mathrm{d}t \bigg). 
\end{align*}
We have 
\begin{align*}
    \mathcal{H}\big(t,z,g,\hat\pi(t,z,g)\big) \leq & -\frac{1}{2}\mathrm{Tr}\Big[\gamma zz^{\top}\Sigma^{\text{\rm obs}}\big(t,\hat\pi(t,z,g)\big)\Sigma(\Sigma^{\text{\rm obs}}\big(t,\hat\pi(t,z,g)\big)^{\top}\big)\Big] + \Big(\mu^{\text{\rm obs}}\big(t,\hat\pi(t,z,g)\big)\Big)_1.
\end{align*}
As $\hat{\pi}(t,z,g)<+\infty$ is uniformly bounded and strictly positive, $\Sigma$ is definite positive and the components of $\Sigma^{\text{\rm obs}}$ are positive, we observe that when $\|z\|_2+\|g\|_2\longrightarrow+\infty$, the first term goes to $-\infty$ while the second term is bounded. Therefore, the supremum on $\mathcal{O}^{\text{obs}}$ cannot be attained for infinite values.

\medskip
If we now choose the incentives $z^{\star,\text{obs}}(t),g^{\star,\text{obs}}(t),\pi^\star\big(t,z^{\star,\text{obs}}(t),g^{\star,\text{obs}}(t)\big)$ as the maximisers of $\mathcal{H}$, they are Borel-measurable deterministic functions of $t\in[0,T]$ thus belong to the set $\mathcal{ZG}^{\text{obs}}$ and are bounded on $[0,T]$, so that
\begin{align*}
    \mathcal{E}\bigg(\int_0^\cdot\Big(\Big(\Sigma^{\text{\rm obs}}\big(t,\pi^\star(t,z^{\star,\text{obs}}(t),g^{\star,\text{obs}}(t))\big)\Big)_{0,:} - z^{\star,\text{obs}}(t)^{\top}\Sigma^{\text{\rm obs}}\Big(t,\pi^\star\big(t,z^{\star,\text{obs}}(t),g^{\star,\text{obs}}(t)\big)\Big)  \Big) \d W_t\bigg)_T
\end{align*}
is a $\mathbb{P}^{\pi^\star}$-martingale and we obtain
\begin{align*}
    \widetilde{V}_0^P = U_P\bigg(\int_0^T  \mathcal{H}\Big(t,z^{\star,\text{obs}}(t),g^{\star,\text{obs}}(t),\pi^\star\big(t,z^{\star,\text{obs}}(t),g^{\star,\text{obs}}(t)\big)\Big)\mathrm{d}t \bigg).
\end{align*}
\end{proof}
Static maximisation \eqref{opt_pb_principal} over $(z,g)\in P$ can easily be handled with classic root-finding algorithms for a large portfolio of green bonds.\footnote{In practice, we observe that for the set of parameters we choose for the numerical experiences, the function $h^{\text{obs}}$ is strictly concave with respect to its last variable thus admits a unique maximizer $\hat\pi$.} Before moving to the numerical experiments, we discuss the form and implementability of the optimal contract.  

\subsection{Discussion}

\subsubsection{On the form of the optimal contracts}

The contract consists of the following elements: 
\begin{itemize}
    \item The term $Z^{\star \text{\rm obs}}_X$ is a compensation given to the investor with respect to the risk associated to the evolution of the portfolio process. If $Z^{\star \text{\rm obs}}_X>0$ (resp. $Z^{\star \text{\rm obs}}_X<0$), the government encourages to increase (resp. decrease) the value of the portfolio: between two times $t_2>t_1$, the investor receives approximately the amount $(Z^{\star \text{\rm obs}}_X)_{t_1}(X_{t_2}-X_{t_1})$.  
    \item For $i\in \{1,\dots,d^g\}$, the term $Z_i$ is a compensation given to the investor with respect to the volatility risk associated to the evolution of the $i$-th green bond price. Between two times $t_2>t_1$, the investor receives approximately the amount $(Z^{\star \text{\rm obs}}_i)_{t_1}(W^i_{t_2}-W^i_{t_1})$: if $Z^{\star \text{\rm obs}}_i$ is close to zero, the government does not give compensation with respect to the volatility of the $i$-th green bond and conversely for $Z^{\star \text{\rm obs}}_i$ far from zero. The intuition behind $Z^{\star \text{\rm obs}}_I$ is the same.
    \item The diagonal terms of $\Gamma^{\text{\rm obs}}$ are compensations with respect to the quadratic variation of the portfolio process and the risk sources of the green bonds and the index. For example if $\Gamma^{\star \text{\rm obs}}_X>0$, the government provides remuneration to the investor for a high quadratic variation (which here can be thought of as volatility) of the portfolio process. If $\Gamma_X<0$, the government penalises a high volatility of the portfolio process. 
    \item The non-diagonal terms of $\Gamma^{\star \text{\rm obs}}$ are compensations with respect to the quadratic covariation of the portfolio process and the risk sources of the green bonds and the index. For example, if $\Gamma^{\star \text{\rm obs}}_{X,i}>0$ for $i\in \{1,\dots,d^g\}$ the government provides remuneration to the investor for similar moves of the portfolio process and the $i$-th green bond. If $\Gamma^{\star \text{\rm obs}}_{X,i}<0$, the government encourages opposite moves of the portfolio process and the $i$-th green bond.
    \item The term $G^{\text{\rm obs}}(t,Z^{\star \text{\rm obs}},\Gamma^{\star \text{\rm obs}})$ is a continuous coupon that is given to the investor. It corresponds to the utility of the investor in the case $\xi=0$.
\end{itemize}

For reasonable choices of parameters $(\alpha,\beta,G)$, the supremum of $h^{\text{obs}}$ and in \eqref{opt_pb_principal} are strictly concave functions so that an optimiser is quickly found using root-finding algorithms. Note that the optimal contract is indexed on the portfolio process $X$ the sources of risk coming from the green bonds $W^g$ and the one coming from the index $W^I$. This can be reformulated as an indexing on $X$ and the prices of the bonds. In this case we define
\begin{align*}
   B^{\text{\rm obs},p} :=
   \begin{pmatrix}
        X \\
        \log(P^g)\\
        \log(P^I)
   \end{pmatrix}, \; 
   B^{\text{\sout{\rm obs}},p} := \log(P^c),
\end{align*}
\begin{align*}
   & \d B_t^{\text{\rm obs},p} := \mu^{\text{\rm obs},p}(t,\pi_t) \mathrm{d}t + \Sigma^{\text{\rm obs},p}(t,\pi_t) \d W_t, \;  \d B_t^{\text{\sout{\rm obs}},p} := \mu^{\text{\sout{\rm obs}},p}(t) \mathrm{d}t+ \Sigma^{\text{\sout{\rm obs}},p}(t) \d W_t,
\end{align*}
where 
\begin{align*}
    & \mu^{\text{\rm obs},p}(t,\pi) := 
    \begin{pmatrix}
        \pi^g \cdot \big(r^g(t)+\eta^g(t)\circ\sigma^g(t)\big) + \pi^c \cdot \big(r^c(t)+\eta^c(t)\circ\sigma^c(t)\big) + \pi^I \mu^I(t) \\
        r^g(t)+\eta^g(t)\circ\sigma^g(t) - (\sigma^g(t))^{\top}\Sigma^g \sigma^g(t) \\
        \mu^I(t) - \frac{\big(\sigma^I(t)\big)^2}{2}
    \end{pmatrix},\\
   & \Sigma^{\text{\rm obs},p}(t,\pi) := 
    \begin{pmatrix}
        (\pi^g\circ\sigma(t)^g)^{\top} &(\pi^c\circ\sigma(t)^c)^{\top} & \pi^I \sigma^I(t) \\
        \text{diag}\big(\sigma^g(t)\big) & \mathbf{0}_{d^g,d^c} & \mathbf{0}_{d^g,1} \\
        \mathbf{0}_{1,d^g} & \mathbf{0}_{1,d^c} & \sigma^I(t)
    \end{pmatrix}, \\
    & \mu^{\text{\sout{\rm obs}},p}(t) := 
    \begin{pmatrix}
    r^c(t)+\eta^c(t)\circ\sigma^c(t) - (\sigma^c(t))^{\top}\Sigma^c \sigma^c(t)
    \end{pmatrix}, \; \Sigma^{\text{\sout{\rm obs}},p}(t) := 
    \begin{pmatrix}
        \mathbf{0}_{d^c,d^g} & \text{diag}\big(\sigma^c(t)\big) & \mathbf{0}_{d^c,1}
    \end{pmatrix}. 
\end{align*}
This leads to minor changes in the computations and the optimal incentives. 
\subsubsection{On the practical implementation of the contract}\label{sec_practical_implementation}

We will show in the numerical section that the processes $(\pi^\star,Z^\star,\Gamma^\star)$ show a rather constant behaviour through the period $[0,T]$. Thus, the optimal contract does not need a frequent re-calibration throughout the year. This suggests the following approximation
\begin{align}\label{approx_contrat}
    \xi^\star \approx \xi^\star_0 + \bar Z^{\star \text{\rm obs}}\cdot B_T^{\text{\rm obs}} + \frac{1}{2}\mathrm{Tr}\Big[(\bar \Gamma^{\star \text{\rm obs}}+ \gamma \bar Z^{\star \text{\rm obs}} (\bar Z^{\star \text{\rm obs}})^{\top}) \langle B^{\text{\rm obs}}\rangle_T\Big] - \int_0^T h^{\text{\rm obs}}\Big(t,\bar Z^{\star \text{\rm obs}},\bar \Gamma^{\star \text{\rm obs}},\pi^\star(t,\bar Z^{\star \text{\rm obs}},\bar \Gamma^{\star \text{\rm obs}})\Big) \mathrm{d}t, 
\end{align}
where $\bar Z^{\star \text{\rm obs}},$ and $\bar \Gamma^{\star \text{\rm obs}}$ are constants corresponding the average of $z^{\star \text{\rm obs}}(t),g^{\star \text{\rm obs}}(t)$ over $[0,T]$ defined by
\begin{align*}
    & \bar Z^{\star \text{\rm obs}}=\big(\bar Z_{X}^{\star \text{\rm obs}}, \bar Z_{1}^{\star \text{\rm obs}},\dots,\bar Z_{d^g}^{\star \text{\rm obs}},\bar Z_{I}^{\star \text{\rm obs}}\big)^{\top} \in \mathbb{R}^{d^g+2}, \; \bar \Gamma^{\star \text{\rm obs}} =
    \begin{pmatrix}
    \bar \Gamma_X^{\star \text{\rm obs}} & \Gamma_{X,1}^{\star \text{\rm obs}}&\dots& \Gamma_{X,d^g}^{\star \text{\rm obs}}& \bar\Gamma_{X,I}^{\star \text{\rm obs}} \\
    \bar\Gamma_{X,1}^{\star \text{\rm obs}} & \bar \Gamma_1^{\star \text{\rm obs}} & \dots & \bar \Gamma_{1,d^g}^{\star \text{\rm obs}} & \bar \Gamma_{1,I}^{\star \text{\rm obs}} \\
     \vdots & \vdots & \ddots & \vdots & \vdots \\
     \vdots & \vdots & \vdots & \ddots & \Gamma_{d^g,I}^{\star \text{\rm obs}}  \\
     \bar \Gamma_{X,I}^{\star \text{\rm obs}} &  \Gamma_{1,I}^{\star \text{\rm obs}}  & \dots & \Gamma_{d^g,I}^{\star \text{\rm obs}} & \Gamma_I^{\star \text{\rm obs}} 
    \end{pmatrix} \in \mathbb{S}_{d^g+2}(\mathbb{R}).
\end{align*}
In order to provide a practical implementation of the contract, we propose a static replication of its payoff using financial instruments. First, note that the incentives $\bar Z_X^{\star \text{\rm obs}},$ and $\bar \Gamma_X^{\star \text{\rm obs}}$ are indexed on the holdings of the investor, thus do not need any replication using financial instruments. The portion $\bar Z^{\star \text{\rm obs}}\cdot B_T^{\text{\rm obs}}$ of the contract can be easily replicated using log-contracts. For example, for $i\in \{1,\dots,d^g\}$, we replicate $\bar Z_{i}^{\star \text{\rm obs}} (B_T^{\text{\rm obs}})_i$ using a long position of size $Z_{i}^{\star \text{\rm obs}}$ on a log-contract on the $i$-th green bond with maturity $T$. In this section, all the derivatives products will have a maturity equal to $T$. 

\medskip
The portion of the contract with respect to quadratic variation and covariation terms are more subtle to replicate. Define the matrix $\tilde{\mathcal{C}}\in \mathbb{S}_{d^g+2}(\mathbb{R})$ whose coefficients are given by
\begin{align*}
    \tilde{\mathcal{C}}_{i,j} := \sum_{k=1}^{d^g+2} \mathcal{C}_{i,k}\langle B_{k,j}^{\text{\rm obs}}\rangle_T, \; \mathcal{C}_{i,j} := \bar \Gamma_{i,j}^{\star \text{\rm obs}} + \gamma \bar Z_i^{\star \text{\rm obs}} \bar Z_j^{\star \text{\rm obs}}, \;  (i,j)\in\{1,\dots, d^g+2\}.
\end{align*}
Then, we can rewrite $ \frac{1}{2}\mathrm{Tr}\Big[\big(\bar \Gamma^{\star \text{\rm obs}}+ \gamma \bar Z^{\star \text{\rm obs}} (\bar Z^{\star \text{\rm obs}})^{\top}\big) \langle B^{\text{\rm obs}}\rangle_T\Big] = \frac{1}{2}\sum_{i=1}^{d^g+2}\tilde{\mathcal{C}}_{i,i}.$ Following the reasoning of \citeauthor*{carr2008robust} \cite{carr2008robust}, we note that the quadratic variations and covariations on the logarithm of the green bonds and the index of conventional bonds can be replicated statically using variance and covariance swaps on the bonds. Finally, the portfolio process is equivalent to holding $\pi^{\star,g}$ green bonds, $\pi^{\star,c}$ conventional bonds and $\pi^{\star,I}$ index. Thus, the quadratic covariation between the portfolio process $X$ and the bonds can be replicated using a linear combination of variance and covariance swaps. 

\medskip
We are now in position to state the replication strategy for the implementation of the contract. The proof is an application of the no-arbitrage principle and It\=o's formula on the logarithm of the bond prices. 
\begin{proposition}
The replication strategy on $[0,T]$ of the optimal contract in \eqref{approx_contrat} is as follow: 
\begin{itemize}
    \item For $i\in \{1,\dots,d^g\}$, hold a position of size $\bar Z_{i}^{\star \text{\rm obs}}$ in a log-contract on the $i$-th green bond. 
    \item Hold a position of size $\bar Z_{I}^{\star \text{\rm obs}}$ in a log-contract on the index of conventional bonds. 
    \item For $i\in \{2,\dots,d^g+1\}$, $k\in \{2,\dots,d^g+1\}$, hold a position of size $\frac{1}{2}\mathcal{C}_{i,k}$ in a covariance swap between the $(i-1)$-th and the $(k-1)$-th green bonds. 
    \item For $i=d^g+2$, $k\in \{2,\dots,d^g+1\}$, hold a position of size $\frac{1}{2}\mathcal{C}_{i,k}$ in a covariance swap between the index of conventional bonds and the $(k-1)$-th green bonds. 
    \item For $i=k=d^g+2$, hold a position of size $\frac{1}{2}\mathcal{C}_{i,i}$ in a variance swap on the index of conventional bonds.    
    \item For $i=1$, $k\in \{2,\dots,d^g+1\}$, $l^g\in \{1,\dots,d^g\},$ $l^c\in \{1,\dots,d^c\}$, hold a position of size $\frac{1}{2}\mathcal{C}_{i,k}\pi_{l^g}^{\star,g}$ in a (co)variance swap between the $(k-1)$-th and the $l^g$-th green bonds, a position of size $\frac{1}{2}\mathcal{C}_{i,k}\pi_{l^c}^{\star,c}$ between the $(k-1)$-th green bond and the $l^c$-th conventional bonds, and a position of size $\frac{1}{2}\mathcal{C}_{i,k}\pi^{\star,I}$ in a covariance swap between the index of conventional bonds and the $(k-1)$-th green bond. 
    \item For $i=1$, $k=d^g+2$, $l^g\in \{1,\dots,d^g\},$ $l^c\in \{1,\dots,d^c\}$, hold a position of size $\frac{1}{2}\mathcal{C}_{i,k}\pi_{l^g}^{\star,g}$ in a covariance swap between the index of conventional bonds and the $l^g$-th green bond, a position of size $\frac{1}{2}\mathcal{C}_{i,k}\pi_{l^c}^{\star,c}$ between the index of conventional bonds and the $l^c$-th conventional bond, and a position of size $\frac{1}{2}\mathcal{C}_{i,k}\pi^{\star,I}$ in a variance swap on the index of conventional bonds. 
\end{itemize}
\end{proposition}
The contract can be implemented practically only by using the value of the portfolio of bonds, log-contracts, variance and covariance swaps on the different bonds. 
\begin{remark}\label{remark_implementation}
We would like to emphasise that, even though it is possible to replicate in practice the optimal contract using variance and covariance swaps on the government bonds, these derivatives might be highly illiquid on financial markets. However, it is possible to replicate these volatility derivatives using the log-contracts and the bonds. Indeed, a variance swap on a bond $P_t$ $($we omit to describe the type of bond for notational simplicity$)$ of maturity $T$ can be replicated by holding for all $t\in [0,T]$ one log-contract that pays $-2\log(P_T/P_0)$ and $2/P_t$ bonds $P_t$. A covariance swap on the bonds $P_t^1$, and $P_t^2$ can be replicated by holding for all $t\in [0,T]$ one log-contract that pays $-2\log(P_T^1/P_0^1)$, one log-contract that pays $-2\log(P_T^2/P_0^2)$, short $\frac{1}{2}$ variance swap on $P^1$, and short $\frac{1}{2}$ variance swap on $P^2$, long $1/(P_t^1 P_t^2)$ bond $P_t^3 := P_t^1 P_t^2$. Thus, the optimal contract $\xi$ in \eqref{opt_contract_fin} can be implemented only using bond prices and log-contracts. 
\end{remark}
Finally, note that if vanilla options on the futures on the bonds are available on the market, one can use the Carr-Madan formula, see \citeauthor*{carr1999option} \cite{carr1999option} to replicate the log-contract payoffs in \Cref{remark_implementation}. Thus, the optimal contract in \eqref{approx_contrat} can be implemented in practice in three different ways: using the bond prices, the portfolio process, the variance and covariance swaps on the bonds; using the bond prices, the portfolio process, and the log-contracts on the bonds; or using the bond prices, the portfolio process, and vanilla options on the bond prices.

\section{Numerical results}\label{sec_numerical}

In the current section, we provide numerical examples illustrating the efficiency of our incentives method. 
\subsection{Data, key results and remarks for the policy-maker}

We illustrate our methodology on an example with real-world data. The dataset is composed of $3$ French governmental bonds, one green bond and two conventional bonds with the following characteristics.

\begin{table}[H]
\begin{center}
\begin{tabular}{c|c|c|c|c|c|c|}
\cline{2-7}

                                           & \small Bloomberg Ticker & \small Valuation date & \small Maturity   &\small Amount issued &\small Issue price &\small Coupon \\ \hline
\multicolumn{1}{|c|}{\small Green bond}           & \small FRTR 1 3/4  &  \small 24/01/2017     & \small 25/06/2039 & \small 27.375b       & \small 100.162     & \small1.75   \\ \hline
\multicolumn{1}{|c|}{\small Conv. bond 1} & \small FRTR 6  &\small 02/01/1994     &\small 25/10/2025 & \small 30.654b       &\small 95.29       & \small 6.     \\ \hline
\multicolumn{1}{|c|}{\small Conv. bond 2} & \small FTRT 4 &\small 09/03/2010     & \small 25/04/2060 & \small16.000b       & \small 96.34       &\small 4.     \\ \hline
\end{tabular}
\end{center}
\end{table}

We also define the index of conventional bonds $I_t$ as a geometric average of the conventional bonds, weighted by the amount issued. We perform the calibration using the daily prices of the bonds from $10/04/2019$ to $10/04/2020$ and the following affine parametrisation for short-term rates, volatilities and risk premiums:
\begin{align*}
    & r^g(t)=a^{r,g}+b^{r,g}(T^g-t),\; \eta^g(t)=a^{\xi,g}+b^{\xi,g}(T^g-t),\; \sigma^g(t)=a^{\sigma,g}+b^{\sigma,g}(T^g-t), \\
    & r^{1,c}(t)=a_1^{r,c}+b_1^{r,c}(T^{1,c}-t),\; \eta^{1,c}(t)=a_1^{\xi,c}+b_1^{\xi,c}(T^{1,c}-t),\; \sigma^{1,c}(t)=a_1^{\sigma,c}+b_1^{\sigma,c}(T^{1,c}-t), \\
    & r^{2,c}(t)=a_2^{r,c}+b_2^{r,c}(T^{2,c}-t),\; \eta^{2,c}(t)=a_2^{\xi,c}+b_2^{\xi,c}(T^{2,c}-t),\; \sigma^{2,c}(t)=a_2^{\sigma,c}+b_2^{\sigma,c}(T^{2,c}-t) \\
    & \mu^I(t) = a^{\mu,I} + b^{\mu,I}(T^I-t), \; \sigma^I(t) = a^{\sigma,I} + b^{\sigma,I}(T^I-t),
\end{align*}
with $T^g=19.73,$ $T^{1,c}=6.06$, $T^{2,c}=40.58$, and $T^I=18.29$. In order to calibrate the dynamics of the bonds in \eqref{dynamics_bonds} over the period, we use a classic least-square algorithm and we obtain the following set of parameters
\begin{align*}
    & a^{r,g}= -0.07, \; b^{r,g} = 0.66, \; a^{\xi,g} = 0.38, \; b^{\xi,g}=0.13, \; a^{\sigma,g}=0.41, \; b^{\sigma,g}=0.31, \\
    & a_1^{r,c}= -0.05,\; b_1^{r,c} = -0.91,\; a_1^{\xi,c}= 0.01,\; b_1^{\xi,c} = 0.30,\; a_1^{\sigma,c}= 0.11,\; b_1^{\sigma,c} = 0.26, \\
    & a_2^{r,c}= 0.28,\; b_2^{r,c} = 0.02, \; a_2^{\xi,c}= 0.12,\; b_2^{\xi,c} = -0.99, \; a_2^{\sigma,c}= 0.10,\; b_2^{\sigma,c} = -0.96,\\
    & a^{\mu,I} = -0.01,\; b^{\mu,I}=0.53,\; a^{\sigma,I} = 0.01,\; b^{\sigma,I}=0.92, 
\end{align*}
and the correlation matrix is given by
\begin{align*}
\Sigma = 
    \begin{pmatrix}
    1 & 0.2 & 0.8 & 0.8 \\
    0.2 & 1 & 0.2 & 0.7 \\
    0.8 & 0.2 & 1 & 0.7 \\
    0.8 & 0.7 & 0.7 & 1
    \end{pmatrix}. 
\end{align*}
The time horizon of the investor and the government is equal to one year, i.e $T=1$. We define a so-called reference case, which is a reference to analyze the impact of our incentives policy. In this setting, 
\begin{align*}
    \nu = \gamma = 1,\;  G=\mathbf{0}_{d^g},\; \kappa=0,\; \beta=(0.4,0.4,0.4,0.4),\;  \alpha=(0.2,0.2,0.3,0.5).
\end{align*}
Thus, the investor and the government have the same risk aversion, and the government has no specific incentives to increase the investments in the green bond. The only objective of the government is to maximise the value of the portfolio of bonds. The investor has the same cost intensity for every bonds and wishes to invest more in the index and the second conventional bond compared to the green and the first conventional bond. This corresponds to a risk-averse investor who prefers a diversified portfolio of conventional bonds, and is reluctant to invest in the green bonds. Finally, the utility reservation of the investor is set equal to the his utility in the case $\xi=0$. 

\medskip
We summarise the important empirical findings coming from the numerical results. 
\begin{itemize}
    \item The methodology we propose outperforms significantly the current tax-incentives policy: for a same result in terms of green investments, our methodology leads to a value of the portfolio process $15\%$ to $20\%$ higher. 
    \item The optimal investment policy is robust to model specification: by using a one-factor model on the short-term rates of the green bond, we observe that the investor's strategy oscillates slightly around the one obtained with deterministic rates. 
    \item The optimal controls show a rather constant behaviour throughout the year: The government does not have to frequently recalibrate the optimal contract. 
    \item The government can increase the amount invested in the green bonds by the mean of $G$ and $\kappa$. This decreases his utility as he must provides higher incentives to the investor.
    \item The most important incentive with respect to the contractible variables is $Z^\star_X$: The government always encourage a higher value of the portfolio of bonds by setting $Z^\star_X >0$. 
    \item When the government provides incentives to increase the investment in green bonds, he encourages higher variations of the value of the portfolio in order to compensate the amount given to the investor.
    \item At the expense of some substantial utility loss, the government can propose a contract indexed only on the contractible variables. This results in a higher incentive $Z^\star_X$. 
\end{itemize}
We also provide some general remarks for the policy-maker.
\begin{itemize}
    \item The parameters $(\alpha,\beta)$ modelling the preferences of the investor should be calibrated using the historical data on the issuance of bonds. For example, for $i\in \{1,\dots,d^g\}$, the coefficient $\alpha_i$ should be equal to the historical amount invested in the green bond $P_i^{g}$, and $\beta_i$ should be equal to the variance of the amount invested in this green bond throughout the year. Note however that one historical data on bonds with the same characteristic may not be available especially for countries with small amounts issued. Thus, the parameters $(\alpha,\beta)$ might be re-scaled depending on the maturity and the coupon of the newly issued bond: A risky investor such as a fixed-income hedge fund might increase his investment in the bond if it offers a higher coupon, whereas institutional investors such as pension funds will tend to buy bonds with a better rating. 
    \item The risk-aversion parameter $\gamma$ should be chosen such that, in the case $\xi=0$ and with $(\alpha,\beta)$ chosen as explained previously, the optimal controls $\pi^\star$ correspond roughly to the historical positions of the investor. 
    \item The risk-aversion parameter $\nu$ should be chosen heuristically such that the optimal contract offered to the investor bring the investments closer to the target $G$ and the amount $\xi^\star$ offered by the government is reasonable. The terms `closer to' and `reasonable' have to be interpreted by the policy-maker in view of their own budget constraints and political objectives. 
    \item In the case of a small number of bonds issued, the government can, for sake of simplicity, propose a contract indexed only on the value of the portfolio. 
\end{itemize}

\subsection{Reference case}

\subsubsection{Optimal controls and comparison with the no-contract case}

In the absence of a contract, that is $\xi=0$, the investor matches his investments $\pi^\star(\xi)$ with the target $\alpha$ as he has no incentives to deviate. Thus, the optimal investments are given by $\pi^{\star g}(0) = 0.2, \;  \pi^{\star c}(0) = (0.2,0.3),\; \pi^{\star I}(0) = 0.5$. We can now analyze the influence of the contract on the behaviour of the investor. We first show in \Cref{fig1} the evolution of the optimal investment policy $\pi^\star$ and the optimal incentives $Z^\star,$ and $\Gamma^\star$ through time. One can see that, even if the risk premia, the short-term rates and volatility processes have a deterministic affine structure with respect to time, the processes $(\pi^\star,Z^\star,\Gamma^\star)$ show a rather constant behaviour through the year. Thus, the optimal contract does not need frequent recalibration through the year.

\begin{figure}[H]
\centering
\begin{minipage}[c]{.46\linewidth}
  \begin{center}
      \includegraphics[width=0.8\textwidth]{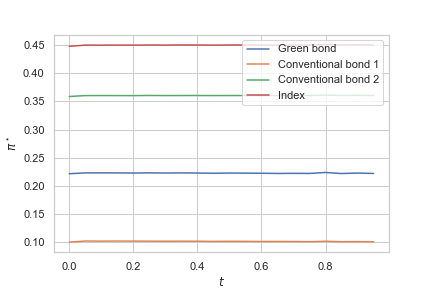}
      \vspace{-3mm}
    \end{center}
\end{minipage} \hfill
\begin{minipage}[c]{.46\linewidth}
   \begin{center}
       \includegraphics[width=0.8\textwidth]{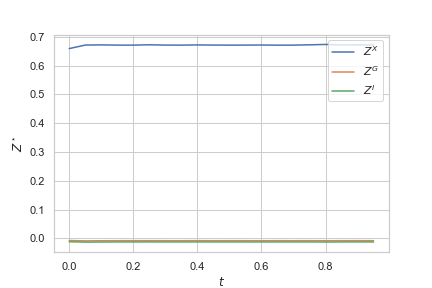}
       \vspace{-3mm}
     \end{center}
  \end{minipage} \\
  \vspace{-3mm}
\begin{minipage}[c]{.46\linewidth}
   \begin{center}
       \includegraphics[width=0.8\textwidth]{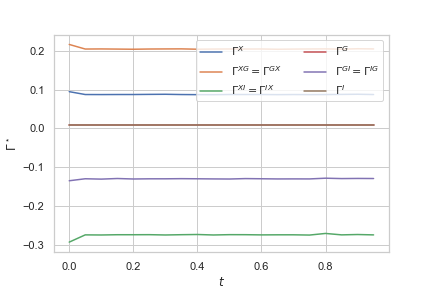}
       \vspace{-3mm}
     \end{center}
  \end{minipage}
\caption{\small Optimal investment policy (upper left), optimal incentives $Z^\star$ (upper right) and $\Gamma^\star$ (bottom) as a function of time.}
\label{fig1}
\end{figure}
Compared to the case $\xi=0$, we observe that the contract increases the investment in the green bond and the second conventional bond, while reducing the investment in the index and the first conventional bond. 
Given the dynamics of the bonds described previously, as well as the preferences of the investor, it is natural that he invests mostly in the index and the second conventional bond. As the green bond has a higher short-term rate and risk premium than the first conventional bond, the traders invests a higher part of his wealth in it.

\medskip
The optimal incentives with respect to the sources of risk is as follow: the incentives with respect to the green bond and the index of conventional bonds are set to zero, whereas the incentive with respect to the value of the portfolio of bonds is strictly positive. Thus, the government provides incentives only to increase the value of the portfolio. We observe at the bottom of \Cref{fig1} the incentives with respect to the quadratic variations of the contractible variables. The government provides no incentives with respect to the quadratic variation of the index and the green bond while it encourages a high quadratic variation of the portfolio process. The incentives with respect to the quadratic covariations are as follow: the government penalises a high covariation between the portfolio process and the index as well as between the green bonds and the index, while encouraging a high covariation between the portfolio and the green bond. 

\subsubsection{Trajectory simulation and portfolio value}

To illustrate the benefits of the use of a contract, we plot in \Cref{fig_compX_01} some simulations of the evolution of the portfolio process over the year with and without contract (that is when $\xi=0$). We observe that the portfolio process is higher when the government provides a contract to the investor. This is also illustrated in Figure \ref{fig_compX_01_bis} where we show the cumulated difference between the portfolio processes with and without contract, using $10000$ simulations.

\begin{figure}[H]
  \begin{center}
      \includegraphics[width=0.8\textwidth]{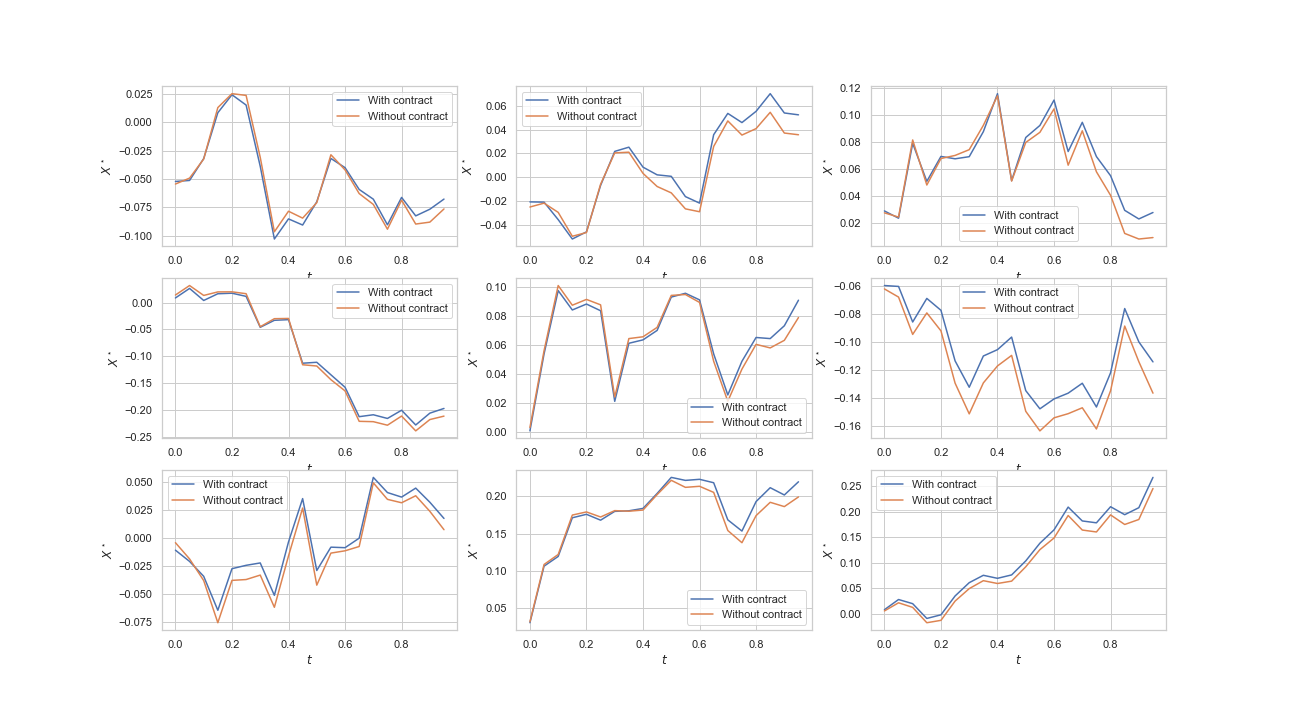}
      \vspace{-3mm}
      \caption{\small Some trajectories of the optimal portfolio process with and without contract.}
      \label{fig_compX_01}  
    \end{center}
\end{figure}

\begin{figure}[H]
  \begin{center}
      \includegraphics[width=9cm,height=5cm]{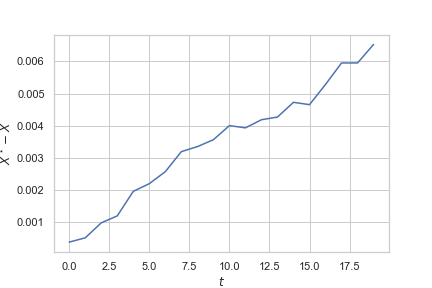}
      \vspace{-3mm}
      \caption{\small Average absolute difference of portfolio value over time, for $10000$ simulations.}
      \label{fig_compX_01_bis} 
    \end{center}
\end{figure}

\subsubsection{Optimal contract with no indexation on quadratic variation}

As the notion of incentives with respect to quadratic variation might not be easy to understand, we present in \Cref{fig2} the optimal investment and incentives when the government set $\Gamma=0$.

\begin{figure}[H]
\centering
\begin{minipage}[c]{.46\linewidth}
  \begin{center}
      \includegraphics[width=0.8\textwidth]{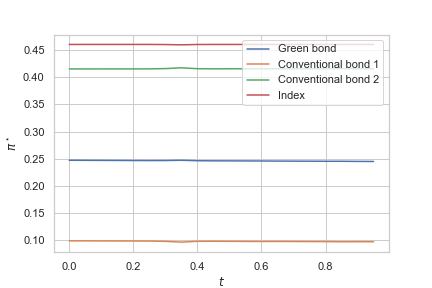}
      \vspace{-3mm}
    \end{center}
\end{minipage} \hfill
\begin{minipage}[c]{.46\linewidth}
   \begin{center}
       \includegraphics[width=0.8\textwidth]{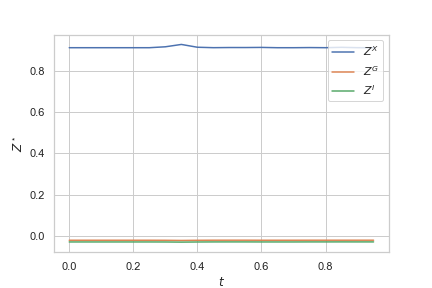}
       \vspace{-3mm}
     \end{center}
  \end{minipage}
\caption{\small Optimal investment policy (left) and optimal incentives $Z^\star$ (right) as a function of time.}
\label{fig2}  
\end{figure}

Compared to \Cref{fig1}, we observe that the government sets a higher incentive on the value of the portfolio, while the optimal investment policy is slightly higher on every asset, but not materially different compared to a framework with an optimal contract depending on both the dynamics and the quadratic variations of the contractible variables. Thus, for sake of simplicity, a government can build an optimal incentives scheme based only on the dynamics of the green bonds, the value of the portfolio and the index of conventional bonds.

\subsubsection{Model robustness}

We show that, using a more complex model for the short-term rates of the green bonds, the results are qualitatively the same. Using the methodology in \Cref{sec_stoch_rates}, we assume that the short-term rate of the green bond is driven by a one-factor stochastic model, that is
\begin{align}\label{taux_sto_ex}
    \d r_t^g = \theta^g(m^g - r_t^g) \d t + \sigma^g \d W_t^{g,r},
\end{align}
where $W^{g,r}$ is a one-dimensional Brownian and $(\theta^g,m^g,\sigma^g)\in \mathbb{R}^3_+$. Using a least-square algorithm, a calibration on the short-term rate curve of the green bond gives the following parameters
\begin{align*}
    \theta^g = 0.4, \; m^g = 0.04, \; \sigma^g = 0.02.
\end{align*}
We show in \Cref{fig_stoc_rates} the optimal investment policy when the short-term rate of the green bond is driven by \eqref{taux_sto_ex}. This is obtained by solving the $4$-dimensional HJB equation \eqref{hjb_stoc_rates} using a fully implicit scheme and locally unidimensional methods on sparse grids.\footnote{In particular, as the bond prices do not vary drastically during the year, we use $10$ time steps, $40$ space steps for the cash process, $10$ for the stochastic rate and $20$ for the risk factors of the green bond and index of conventional bonds.} Note that the optimisation is much harder to complete since for every $\pi^{\star}(t,z,g,r^g)$ we have to solve a $4$-dimensional HJB equation and iterate until we find the optima $(z^\star,\Gamma^\star)$. We observe that the optimal policy oscillates around the values obtained in the case of deterministic short-term rates in \Cref{fig1}. As the bonds are all positively but not perfectly correlated, a change of investment in the green bond induces a change of smaller magnitude in the other bonds. The magnitude of oscillation around the value with deterministic rates is not high, thus we observe same results from a qualitative point of view. As the use of stochastic rates can only be viable for a small portfolio of bonds, and as the difference of behaviour is negligible, we can argue that the use of deterministic short-term rates is more suited to practical applications. 

\begin{figure}[H]
  \begin{center}
      \includegraphics[width=9cm,height=5cm]{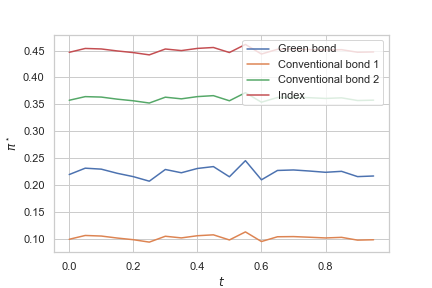}
      \vspace{-3mm}
      \caption{\small Optimal investment policy with stochastic rates.}
      \label{fig_stoc_rates} 
    \end{center}
\end{figure}

\subsubsection{Comparison with current tax-incentives policy}

The purpose of the paper is to show that a form of incentives based on the value of the portfolio and the prices of the bonds performs better than the current tax-incentives policy. As stated in the introduction, the incentives policy to increase investment in green bonds takes the form of tax credit or cash rebate, depending on the amount invested. Thus, in our Principal-Agent framework, it takes the following form
\begin{align*}
    \xi = c\int_0^T \sum_{i=1}^{d^g} \pi_t^g \d t,
\end{align*}
where $c>0$ is the amount of cash rebate or tax credit, controlled by the government. We choose $c$ so that the amount invested in green bonds is the same as in \Cref{fig1}. In \Cref{fig_comp_refcase_avg}, we plot the average relative difference between the cash processes of the government using our optimal policy and the actual tax-incentives. We observe that the difference increases with time, thus for a same result in terms of green investments our optimal contract increases its utility compared to the actual incentives policy. 

\begin{figure}[H]
  \begin{center}
      \includegraphics[width=9cm,height=5cm]{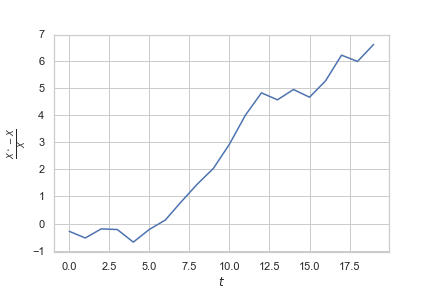}
      \vspace{-3mm}
      \caption{\small Average relative difference (in \%) of portfolio value over time for $10000$ simulations.}
      \label{fig_comp_refcase_avg} 
    \end{center}
\end{figure}

We also show in \Cref{fig_comp_refcase} some trajectories of the value of the portfolio process with the optimal contract and the optimal policy. We observe that the value of the portfolio process is always (slightly) higher in the presence of the optimal contract. In the next subsection we show that when the government wants to achieve a specific target in green investments, the difference between the two policies becomes larger

\begin{figure}[H]
  \begin{center}
      \includegraphics[width=0.8\textwidth]{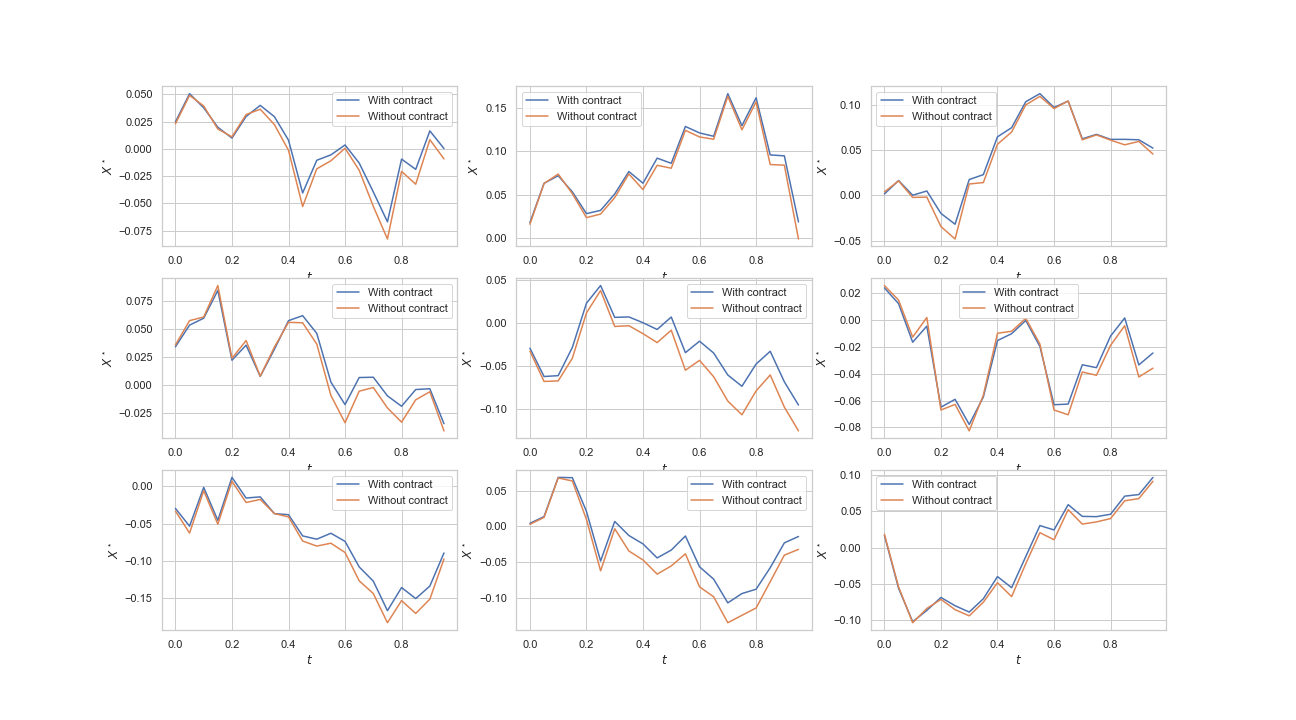}
      \vspace{-3mm}
      \caption{\small Some trajectories of the optimal portfolio process with the optimal contract and with the tax-incentives policy (labeled `without contract').}
      \label{fig_comp_refcase} 
    \end{center}
\end{figure}

\subsection{Influence of the green target}

\subsubsection{Comparison with the reference case}

We now study the impact of the incentives policy we propose when the government seeks to achieve a specific investment target in the green bond.  We take $G=3,$ $\kappa=0.8$ and present in \Cref{fig3} the new optimal controls of the investor and the government.

\begin{figure}[H]
\centering
\begin{minipage}[c]{.46\linewidth}
  \begin{center}
      \includegraphics[width=0.8\textwidth]{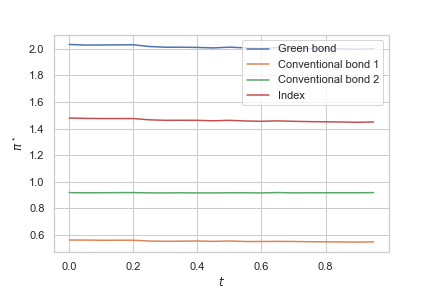}
      \vspace{-3mm}
    \end{center}
\end{minipage} \hfill
\begin{minipage}[c]{.46\linewidth}
   \begin{center}
       \includegraphics[width=0.8\textwidth]{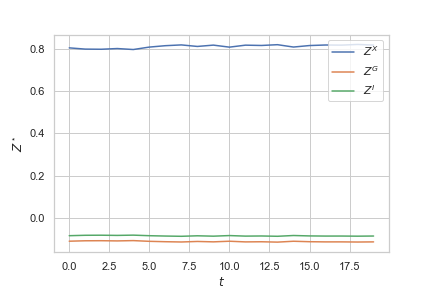}
       \vspace{-3mm}
     \end{center}
  \end{minipage} \\
  \vspace{-3mm}
\begin{minipage}[c]{.46\linewidth}
   \begin{center}
       \includegraphics[width=0.8\textwidth]{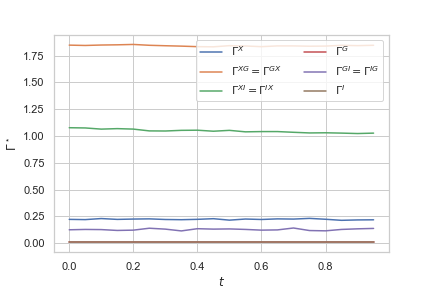}
       \vspace{-3mm}
     \end{center}
  \end{minipage}
\caption{\small Optimal investment policy (upper left), optimal incentives $Z^\star$ (upper right) and $\Gamma^\star$ (bottom) as a function of time.}
\label{fig3}
\end{figure}

The behaviour of the investor is drastically different compared to \Cref{fig1}. He now invests mostly in the green bond, while increasing the amount invested in the other assets. This comes from the fact that all assets are positively correlated so that the additional amount invested in the index is higher than the one invested in the first conventional bond. The government sets a higher incentive with respect to the value of the portfolio. The incentives with respect to the quadratic variation are now all positive and higher than in \Cref{fig1}. While $\Gamma_G$ and $\Gamma_I$ are still set to zero, the incentive with respect to every covariations are now positive. In particular, $\Gamma_{XI}$ has changed from $-0.3$ to $1$ meaning that the government encourages a higher quadratic covariation between the portfolio and the index of conventional bonds. So as to maximise the value of the portfolio while giving higher incentives, the government encourages a higher variance of the portfolio process and positive co-variations between the portfolio and the bond prices. 

\medskip
Note that, while the amount invested in the green bond is higher but not equal to the target of the government. As $\alpha^g = 0.2$, the government has to provide higher incentives to force the investor to shift his preferences toward a much higher investment in the green bond. As in the reference case, we show in \Cref{fig_compX_02_bis} some simulations of the evolution of the portfolio process compared to the case without contract. We observe that the higher investment in green bonds leads to a higher average value of the portfolio process. Moreover due to the higher incentives on the quadratic variations, the portfolio process with the contract is more volatile, as it can be seen in \Cref{fig_compX_02}

\begin{figure}[H]
  \begin{center}
      \includegraphics[width=8cm,height=4cm]{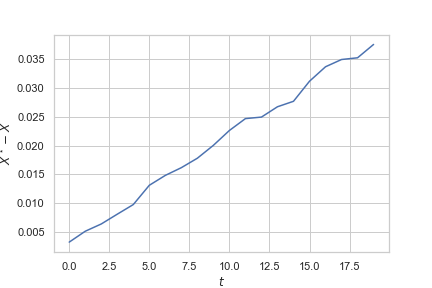}
      \vspace{-3mm}
      \caption{\small Average difference of portfolio value over time, for $10000$ simulations.}
      \label{fig_compX_02_bis} 
    \end{center}
\end{figure}

\begin{figure}[H]
  \begin{center}
      \includegraphics[width=0.8\textwidth]{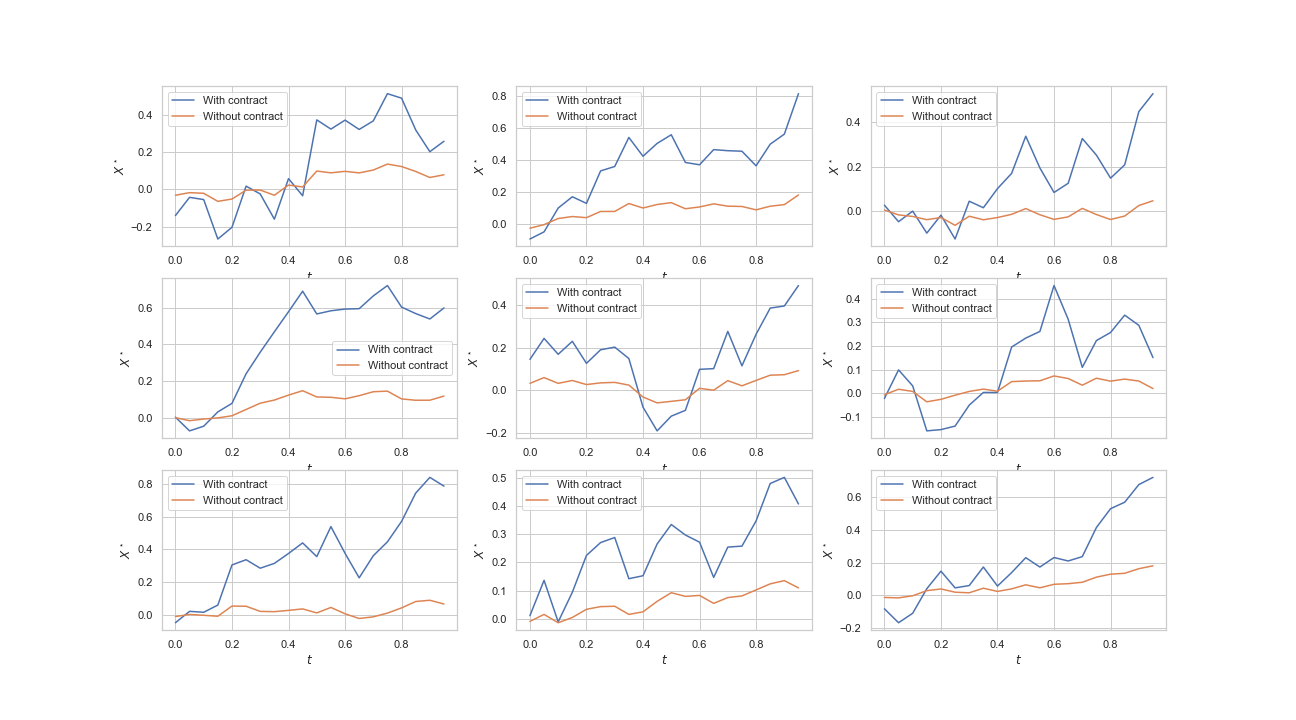}
      \vspace{-3mm}
      \caption{\small Some trajectories of the optimal portfolio process with and without contract.}
      \label{fig_compX_02}  
    \end{center}
\end{figure}

\subsubsection{Comparison with the tax-incentives policy}

We have seen in Figures \ref{fig_comp_refcase_avg} and \ref{fig_comp_refcase} that without specific target in green investments, the optimal contract we propose leads to a higher value of the portfolio process compared to the tax-incentives policy. Here, we set the tax-incentives $c$ so that the investor matches the investment in green bonds obtained with the optimal contract in \Cref{fig3}. We plot in \Cref{fig_comp03sim}, and \Cref{fig_comp03_bis} some trajectories and the average relative difference of cash processes obtained with the optimal contract and the tax-incentives policy. 

\medskip
In this case, the relative differences of value are much higher compared to \Cref{fig_comp_refcase_avg}, and \Cref{fig_comp_refcase}. Thus, if the government has a specific investment target in green bonds, the use of the optimal contract we propose guarantees a much higher value of the portfolio for a similar result than the tax-incentives policy. 

\begin{figure}[H]
  \begin{center}
      \includegraphics[width=0.8\textwidth]{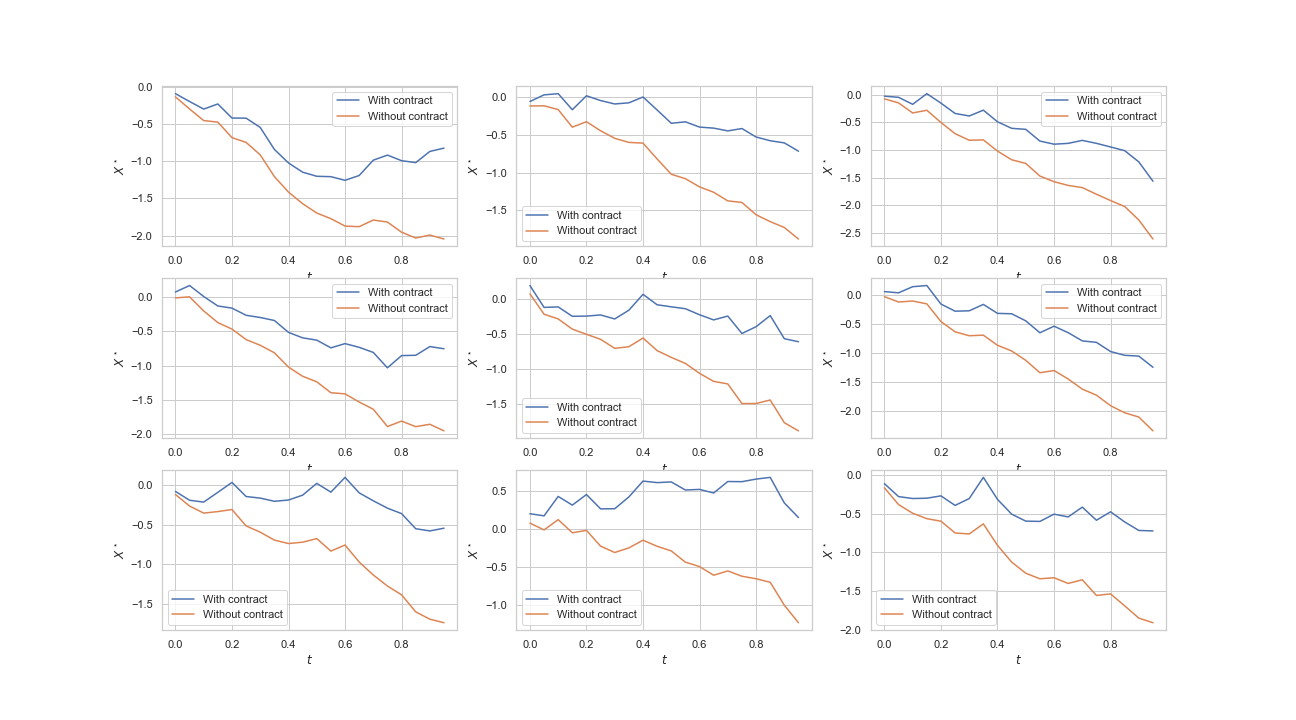}
      \vspace{-3mm}
      \caption{\small Some trajectories of the optimal portfolio process with the optimal contract and with the tax-incentives policy (labeled `without contract').}
      \label{fig_comp03sim}  
    \end{center}
\end{figure}

\begin{figure}[H]
  \begin{center}
      \includegraphics[width=9cm,height=5cm]{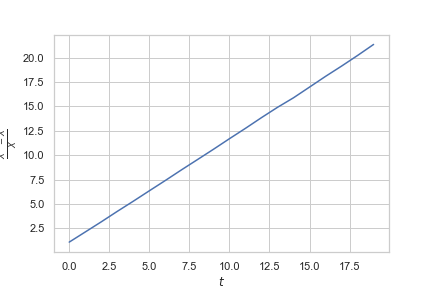}
      \vspace{-3mm}
      \caption{\small Average relative difference (in \%) of portfolio value over time (with the optimal contract and tax-incentives policy), for $10000$ simulations.}
      \label{fig_comp03_bis} 
    \end{center}
\end{figure}

\subsection{Sensitivity analysis}

\subsubsection{Influence of \texorpdfstring{$G$}{G}, and \texorpdfstring{$\kappa$}{k}}

In \Cref{fig4}, we show that reducing the value of $\kappa$ makes the government target harder to achieve. 

\begin{figure}[H]
\centering
\begin{minipage}[c]{.46\linewidth}
  \begin{center}
      \includegraphics[width=0.8\textwidth]{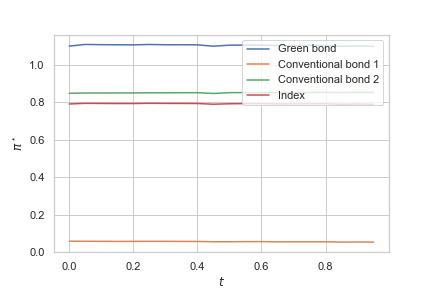}
      \vspace{-3mm}
    \end{center}
\end{minipage} \hfill
\begin{minipage}[c]{.46\linewidth}
   \begin{center}
       \includegraphics[width=0.8\textwidth]{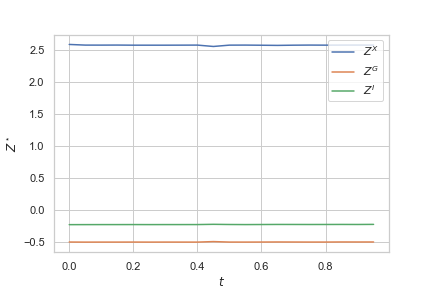}
       \vspace{-3mm}
     \end{center}
  \end{minipage} \\
  \vspace{-3mm}
\begin{minipage}[c]{.46\linewidth}
   \begin{center}
       \includegraphics[width=0.8\textwidth]{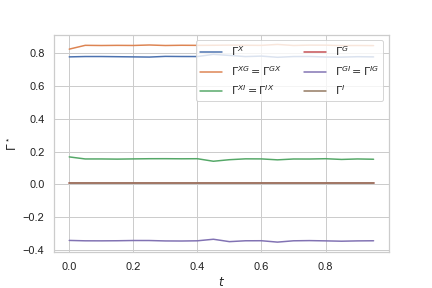}
       \vspace{-3mm}
     \end{center}
  \end{minipage}
\caption{\small Optimal investment policy (upper left), optimal incentives $Z^\star$ (upper right) and $\Gamma^\star$ (bottom) as a function of time.}
\label{fig4}
\end{figure}
In particular, we observe that the amount invested in all the assets has been reduced and especially the amount invested in the green bond. In this case, the government proposes a much higher incentive with respect to the dynamics of the portfolio compared to \Cref{fig3}: as the investment target $G$ is less important (because of a lower $\kappa$, he aims at maximising the value of the portfolio  Moreover, a high quadratic covariation between the green bond and the index is now penalised, while a high variance of the portfolio is encouraged in order to maximise its value.  

\medskip
In \Cref{fig5}, we show that with the parameters $\kappa=0.8$, $G=1$, the investment target of the government can be reached more easily. In this case, the trader invest roughly the same amount in the the green bond and the second conventional bond. The government increases the incentive corresponding to the value of the portfolio compared to \Cref{fig3}. Moreover, he encourages a high variance of the portfolio process while keeping the incentives $\Gamma_G$, $\Gamma_I$ equal to zero. 

\begin{figure}[H]
\centering
\begin{minipage}[c]{.46\linewidth}
  \begin{center}
      \includegraphics[width=0.8\textwidth]{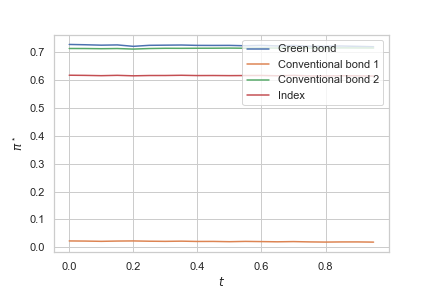}
      \vspace{-3mm}
    \end{center}
\end{minipage} \hfill
\begin{minipage}[c]{.46\linewidth}
   \begin{center}
       \includegraphics[width=0.8\textwidth]{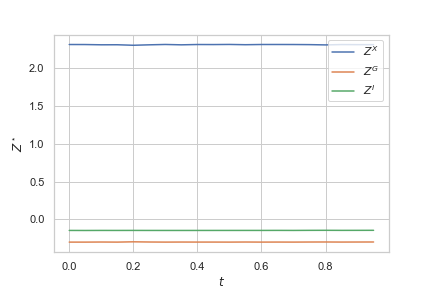}
       \vspace{-3mm}
     \end{center}
  \end{minipage} \\
  \vspace{-3mm}
\begin{minipage}[c]{.46\linewidth}
   \begin{center}
       \includegraphics[width=0.8\textwidth]{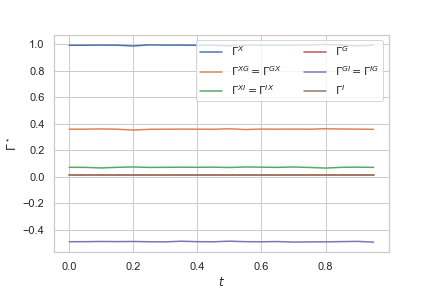}
       \vspace{-3mm}
     \end{center}
  \end{minipage}
\caption{\small Optimal investment policy (upper left), optimal incentives $Z^\star$ (upper right) and $\Gamma^\star$ (bottom) as a function of time}
\label{fig5}
\end{figure}

\subsubsection{Influence of \texorpdfstring{$\alpha$}{a} and \texorpdfstring{$\beta$}{b}}

We studied in the previous section the influence of the government's parameters, that is the target $G$ and the cost intensity $\kappa$. We now show the influence of the targets $\alpha^g,\alpha^c,\alpha^I$ and the cost intensities $\beta^g,\beta^c,\beta^I$ of the investor. In \Cref{fig6}, we place ourselves in the context of the reference case of \Cref{fig1}, except that we set $\alpha^g=0$. This means that the investor is not willing to put money in the green bond. Compared to \Cref{fig1}, we see that in the absence of specific incentives for green investing, the investor effectively sets $\pi^g$ equal to zero. 

\medskip
The other investment policies are slightly changed, as there is now more investment in the second conventional bond than in the index. As neither the government nor the investor are interested in the green bond, the  government provides higher incentives $Z_X$ in order to maximise the value of the portfolio. The incentive $\Gamma_{XG}$ become negative while $\Gamma_X$ becomes positive meaning that the government encourages opposite moves between the price of the green bond and the portfolio process. Moreover, $\Gamma_{XI}$ becomes positive: the government encourages similar moves between the price of the index and the portfolio process. Finally, the incentives corresponding to the quadratic variation of the green bond and the index remain equal to zero. 

\begin{figure}[H]
\centering
\begin{minipage}[c]{.46\linewidth}
  \begin{center}
      \includegraphics[width=0.8\textwidth]{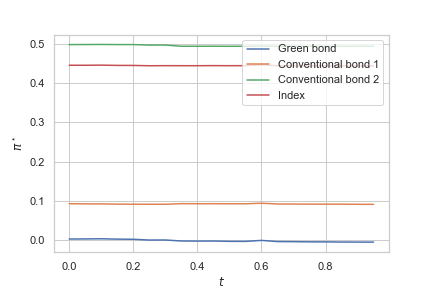}
      \vspace{-3mm}
    \end{center}
\end{minipage} \hfill
\begin{minipage}[c]{.46\linewidth}
   \begin{center}
       \includegraphics[width=0.8\textwidth]{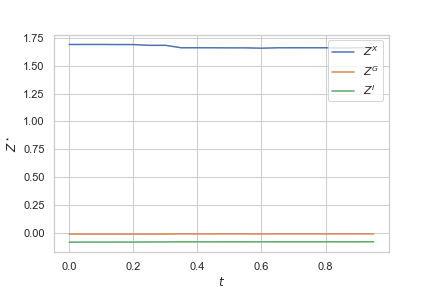}
       \vspace{-3mm}
     \end{center}
  \end{minipage} \\
  \vspace{-3mm}
\begin{minipage}[c]{.46\linewidth}
   \begin{center}
       \includegraphics[width=0.8\textwidth]{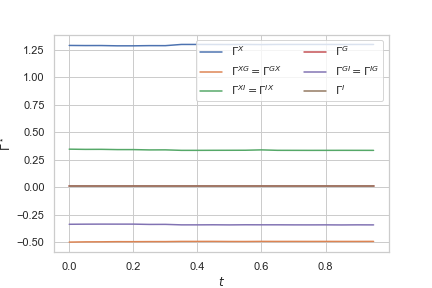}
       \vspace{-3mm}
     \end{center}
  \end{minipage}
\caption{\small Optimal investment policy (upper left), optimal incentives $Z^\star$ (upper right) and $\Gamma^\star$ (bottom) as a function of time.}
\label{fig6}
\end{figure}

In \Cref{fig7}, we compare these results with the case $G=3,\kappa=0.8$ in order to show the influence of the contract when the investor and the government have very different investment targets. We observe that the amount invested in the green bond is clearly higher than in Figure \ref{fig6} where the investor has $\alpha^g=0$ but lower than in Figure \ref{fig3} where the investor has $\alpha^g=0.2$. The incentives with respect to the quadratic variations become positive meaning that the government encourages similar moves of all the contractible variables. In particular, compared to Figure \ref{fig6}, the government gives higher incentives toward similar moves of the portfolio value and the green bond. 

\medskip
We conclude this section by showing in \Cref{fig8} the influence of the cost intensity. We take the same parameters as in \Cref{fig7} except that we set $\beta^g=0.5$. As the intensity cost for moving the green bond target of the investor is higher than in \Cref{fig7}, the optimal investment policy in the green bond is lower. The government sets a higher incentive $Z_X$ to encourage a higher value of the portfolio. The incentives with respect to quadratic variations are materially different compared to Figure \ref{fig7}. In particular, the government encourages opposite moves between the green bond and the index. 

\begin{figure}[H]
\centering
\begin{minipage}[c]{.46\linewidth}
  \begin{center}
      \includegraphics[width=0.8\textwidth]{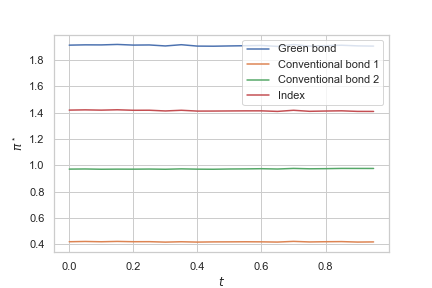}
      \vspace{-3mm}
    \end{center}
\end{minipage} \hfill
\begin{minipage}[c]{.46\linewidth}
   \begin{center}
       \includegraphics[width=0.8\textwidth]{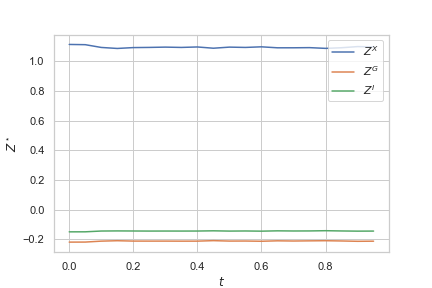}
       \vspace{-3mm}
     \end{center}
  \end{minipage} \\
  \vspace{-3mm}
\begin{minipage}[c]{.46\linewidth}
   \begin{center}
       \includegraphics[width=0.8\textwidth]{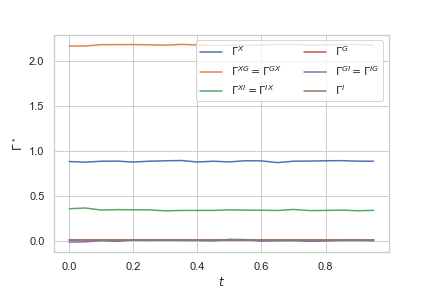}
       \vspace{-3mm}
     \end{center}
  \end{minipage}
\caption{Optimal investment policy (upper left), optimal incentives $Z^\star$ (upper right) and $\Gamma^\star$ (bottom) as a function of time}
\label{fig7}
\end{figure}

\begin{figure}[H]
\centering
\begin{minipage}[c]{.46\linewidth}
  \begin{center}
      \includegraphics[width=0.8\textwidth]{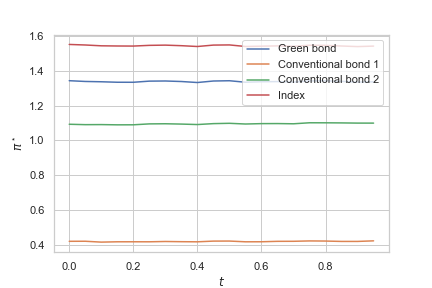}
      \vspace{-3mm}
    \end{center}
\end{minipage} \hfill
\begin{minipage}[c]{.46\linewidth}
   \begin{center}
       \includegraphics[width=0.8\textwidth]{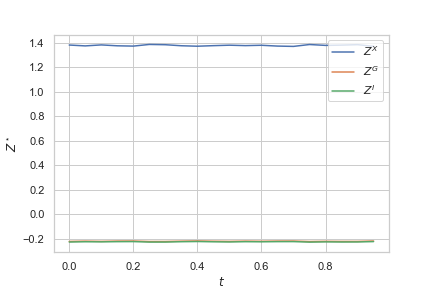}
       \vspace{-3mm}
     \end{center}
  \end{minipage} \\
  \vspace{-3mm}
\begin{minipage}[c]{.46\linewidth}
   \begin{center}
       \includegraphics[width=0.8\textwidth]{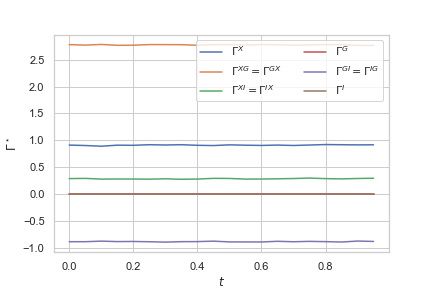}
       \vspace{-3mm}
     \end{center}
  \end{minipage}
\caption{Optimal investment policy (upper left), optimal incentives $Z^\star$ (upper right) and $\Gamma^\star$ (bottom) as a function of time}
\label{fig8}
\end{figure}

\begin{appendix}

\section{Weak formulation of the problem}\label{sec_weak_formulation}
We work on the canonical space $\mathcal{Q}$ of continuous functions on $[0,T]$ with Borel algebra $\mathcal{F}$. The $(d^g+d^c+2)$-dimensional
canonical process is 
\begin{align*}
    \mathcal{B:}= 
    \begin{pmatrix}
    X \\
    W^g \\
    W^c \\
    W^I
    \end{pmatrix}
\end{align*}
and $\mathbb{F}=(\mathcal{F}_t)_{t\in [0,T]}$ is its natural filtration. We define $\mathbb{P}_0$ as the $(d^g+d^c+1)$-dimensional Wiener measure on $\mathcal{Q}$. Thus, $\mathcal{B}$ is a $(d^g+d^c+2)$-dimensional Brownian motion where $(W^g,W^c,W^I)$ has a correlation matrix $\Sigma$ under $\mathbb{P}_0$. We also define $\mathcal{M}(\Omega)$ as the set of probability measures on $(\mathcal{Q},\mathcal{F}_T)$ and  
\begin{align*}
    \mathbb{H}^2(\mathbb{P}_0):=\bigg\{(\pi_t)_{t\in[0,T]}:B\text{\rm -valued}, \mathbb{F}\text{-predictable processes such that }  \mathbb{E}^{\mathbb{P}}\bigg[\int_0^T\|\pi_t\|_2^2 \d t \bigg] <+\infty\bigg\}. 
\end{align*}
We consider the following family of processes, indexed by $\pi\in\mathbb{H}^2(\mathbb{P}_0)$
\begin{align*}
    \mathcal{X}_t^\pi := 
    \begin{pmatrix}
    \int_0^t \Sigma^{\text{\rm obs}}(s,\pi_s) \d \mathcal{B}_s \\
    \int_0^t \Sigma^{\text{\sout{\rm obs}}} \d \mathcal{B}_s
    \end{pmatrix},
\end{align*}
and define the set $\mathcal{P}_m$ as the set of probability measures $\mathbb{P}^\pi \in \mathcal{M}(\mathcal{Q})$ of the form
\begin{align*}
    \mathbb{P}^\pi = \mathbb{P}_0 \circ (\mathcal{X}^\pi)^{-1}, \; \text{for all } \pi\in \mathbb{H}^2(\mathbb{P}_0).
\end{align*}
Thanks to \citeauthor*{bichteler1981stochastic} \cite{bichteler1981stochastic}, we can define a pathwise version of the quadratic variation process $\langle \mathcal{B}\rangle$ and of its density process with respect to the Lebesgue measure
$\hat{\alpha}_t := \frac{d\langle \mathcal{B}\rangle_t}{\mathrm{d}t}$. As the processes $\pi\in\mathcal{A}\subset\mathbb{H}^2(\mathbb{P}_0)$ have all their coordinates strictly positive, the volatility of $\mathcal{B}$ is invertible, which implies in particular that the process $W_t = \int_0^t \hat{\alpha}_s^{-\frac{1}{2}}\d\mathcal{B}_s$
is an $\mathbb{R}^{d^g+d^c+2}$-valued, $\mathbb{P}$-Brownian motion with correlation matrix $\Sigma$ for every $\mathbb{P}\in \mathcal{P}_m$. According to \citeauthor*{soner2013dual} \cite{soner2013dual}, there exists an $\mathbb{F}^{\mathcal{B}}$-progressively measurable mapping $\beta_\pi: [0,T]\times \mathcal{Q} \longrightarrow \mathbb{R}^{d^g+d^c+2}$ such that
\begin{align*}
    \mathcal{B} = \beta_\pi(\mathcal{X}^\pi),\; \mathbb{P}_0\text{-a.s}, \; W = \beta_\pi(\mathcal{B}),\; \mathbb{P}^\pi\text{-a.s}, \; \hat{\alpha}(\mathcal{B})= \pi\big(\beta_\pi(\mathcal{B})\big),  \; \mathrm{d}t \otimes \d\mathbb{P}^\pi\text{-a.e.}
\end{align*}
In particular, the canonical process $\mathcal{B}$ admits the following dynamics for all $\pi\in\mathcal{A}$
\begin{align*}
    \mathcal{B}_t = 
    \begin{pmatrix}
    \int_0^t \Sigma^{\text{\rm obs}}(s,\pi(W_\cdot)) \d W_s \\
    \int_0^t \Sigma^{\text{\sout{\rm obs}}} \d W_s
    \end{pmatrix}, \; \mathbb{P}^{\pi}\text{-a.s}.
\end{align*}
The first coordinate of the canonical process is the desired output process, the $d^g$ next coordinates are the contractible sources of risk, that is the $d^g$ green bonds and the index of conventional bond, and the last $d^c$ coordinates are the non-contractible sources of risk. Then, we can introduce easily the drift of the output process by the means of Girsanov theorem. Denote 
\begin{align*}
    \frac{\d\mathbb{Q}}{\d\mathbb{P}^\pi} := \mathcal{E}\bigg(\int_0^\cdot \tilde{\Sigma}(s) \d W_s \bigg)_T,
\end{align*}
a change of measure independent of the control process $\pi$,
where $\tilde{\Sigma}: [0,T]\longrightarrow \mathcal{M}_{d^g+d^c+2}(\mathbb{R})$ is such that
\begin{align*}
    \tilde{\Sigma}(t) := 
    \begin{pmatrix}
     \big(\frac{r^g(t)+\eta^g(t)\circ\sigma^g(t)}{\sigma^g(t)}\big)^{\top} &  \big(\frac{r^c(t)+\eta^c(t)\circ\sigma^c(t)}{\sigma^c(t)}\big)^{\top} & \frac{\mu^I(t)}{\sigma^I(t)} \\
     \mathbf{0}_{d^g+d^c+1,d^g} & \mathbf{0}_{d^g+d^c+1,d^c} &\mathbf{0}_{d^g+d^c+1,1}
    \end{pmatrix}.
\end{align*}
We finally obtain the desired dynamics for the output process and the $d^g+d^c+1$ sources of risk. 

\section{Proof of Theorem \ref{thm_1_admissible_contract}}
We can define the functions $\sigma:[0,T]\times K \longrightarrow \mathcal{M}_{d^g+d^c+2,d^g+d^c+1}(\mathbb{R}), \lambda:[0,T]\longrightarrow \mathbb{R}^{d^g+d^c+1}$ such that the set of contractible variables $(B_t)_{t\in[0,T]}$ can be rewritten for all $\pi\in\mathcal{A}$ as
\begin{align}\label{controlled_output}
    \d B_t = \sigma(t,\pi_t) \big( \lambda(t) \d t + \d W_t\big),
\end{align}
where for all $(t,p)\in [0,T]\times K$,
\begin{align*}
    \sigma(t,p) := 
        \begin{pmatrix}
        \big(p^g \sigma^g(t)\big)^\top  & \big(p^c \sigma^c(t)\big)^{\top} & p^I \sigma^I(t) \\
        \text{diag}(\sigma^g(t)) & \mathbf{0}_{d^g,d^c} & \mathbf{0}_{d^g,1} \\
        \mathbf{0}_{1,d^g}& \mathbf{0}_{1,d^c} & \sigma^I(t) \\
        \mathbf{0}_{d^c,d^g} & \text{diag}(\sigma^c(t)) & \mathbf{0}_{d^c,1}
    \end{pmatrix},
    \lambda(t) := 
    \begin{pmatrix}
        \big(\frac{r^g(t)+\eta^g(t)\circ\sigma^g(t)}{\sigma^g(t)}\big)^{\top} &  \big(\frac{r^c(t)+\eta^c(t)\circ\sigma^c(t)}{\sigma^c(t)}\big)^{\top} & \frac{\mu^I(t)}{\sigma^I(t)}
    \end{pmatrix}^\top,
\end{align*}
Thanks to \Cref{assumption_bounded} and the definition of $\mathcal{A}$, the functions $\sigma$, and $\lambda$ are bounded. As the function $\sigma(t,\pi)$ is continuous in time for some constant control process $\pi\in\mathcal{A}$, there always exists a weak solution to \eqref{controlled_output}. Thanks to the boundedness of the function $\lambda$, we can use Girsanov's theorem which guarantees that every $\pi\in\mathcal{A}$ induces a weak solution for 
\begin{align*}
    B_t = B_0 + \int_0^t \sigma(s,\pi_s) \d W^{\prime}_s, \; \frac{\d\mathbb{P}^\prime}{\d\mathbb{P}}\bigg|_{\mathcal{F}_T} = \mathcal{E}\bigg(\int_0^\cdot \lambda(s)\cdot \d W_s \bigg)_T,
\end{align*}
where $W^{\prime}$ is a $\mathbb{P}^{\prime}$-Brownian motion. 

\medskip
The cost function $k: K\to \mathbb{R}$ is measurable and bounded by boundedness of the elements of $K$. We introduce the norms 
\begin{align*}
    \|Z^e\|_{\mathbb{H}^p}^p = \sup_{\pi\in\mathcal{A}}\mathbb{E}^{\pi}\bigg[\bigg( \int_0^T\Big|\tilde{\sigma}(t,\pi_t) Z_t\Big|^2 \d t\bigg)^{p/2}\bigg], \quad \|Y^e\|_{\mathbb{D}^p}^p = \sup_{\pi\in\mathcal{A}}\mathbb{E}^{\pi}\bigg[\sup_{t\in[0,T]} |Y_t|^p\bigg],
\end{align*}
for any $\mathbb{F}$-predictable, $\mathbb{R}^{d^g+d^c+2}$-valued process $Z^e$ and $\mathbb{R}$-valued process $Y^e$, and for all $(t,p)\in[0,T]\times K$ $\widetilde{\sigma}:[0,T]\times K\to \mathcal{M}_{d^g+d^c+2}(\mathbb{R})$ is such that
\begin{align*}
    \widetilde{\sigma}^2(t,\pi_t) = \sigma(t,p) \sigma^\top(t,p).
\end{align*}
We also define the functions $H^e: [0,T]\times\mathbb{R}^{d^g+d^c+2}\times\mathbb{S}_{d^g+d^c+2}(\mathbb{R})\times\mathbb{R}\longrightarrow \mathbb{R}$ and $h^e: [0,T]\times\mathbb{R}^{d^g+d^c+2}\times\mathbb{S}_{d^g+d^c+2}(\mathbb{R})\times\mathbb{R}\times K \longrightarrow \mathbb{R}$ as 
\begin{align*}
    & H^e(t,z,g,y):=\sup_{p\in K}h^e(t,z,g,y,p) \\
    & h^e(t,z,g,y,p):= -\gamma k(p)y + z\cdot \sigma(t,p)\lambda(t) + \frac{1}{2}\mathrm{Tr}\Big[g\sigma(t,p)\Sigma(\sigma(t,p))^\top\Big].
\end{align*}
We introduce the set of so-called admissible incentives $\mathcal{ZG}^e$ as the set of $\mathbb{F}$-predictable processes $(Z^e,\Gamma^e)$ valued in $\mathbb{R}^{d^g+d^c+2}\times\mathbb{S}_{d^g+d^c+2}(\mathbb{R})$ such that
\begin{align}\label{cond_integrability}
    \|Z^e\|_{\mathbb{H}^p}^p + \|Y^{e,Z^e,\Gamma^e}\|_{\mathbb{D}^p}^p <+\infty,
\end{align}
for some $p>1$ where for $y^e_0\in\mathbb{R}$,
\begin{align*}
    Y_t^{e, y_0^e, Z^e,\Gamma^e} := y_0^e + \int_0^t Z^e_s \d B_s + \frac{1}{2}\mathrm{Tr}\Big[\Gamma^e_s \d\langle B\rangle_s \Big] - H^e\big(s,Z_s^e,\Gamma_s^e,Y_s^{e, Z^e,\Gamma^e}\big) \d s.
\end{align*}
Condition \eqref{cond_integrability} guarantees that the process $(Y_t^{e, y_0^e, Z^e,\Gamma^e})_{t\in[0,T]}$ is well defined: provided that the right-hand side integrals are well defined, and by noting that $H^e$ is Lipschitz in its last variable (since the cost function $k$ is bounded), $(Y_t^{e, y_0^e, Z^e,\Gamma^e})_{t\in[0,T]}$ is the unique solution of an ODE with random coefficient. Moreover, as $K$ is a compact set and $h^e$ is continuous with respect to its last variable, the supremum with respect to $p$ is always attained. As $(Z^e,\Gamma^e)=(\mathbf{0}_{d^g+d^c+2},\mathbf{0}_{d^g+d^c+2,d^g+d^c+2})\in \mathcal{ZG}^e$, this set is non-empty and we are in the setting of \citeauthor*{cvitanic2018dynamic} \cite{cvitanic2018dynamic}. Using \cite[Proposition 3.3 and Theorem 3.6]{cvitanic2018dynamic}, we obtain that without reducing the utility of the Principal, any admissible contract admits the representation
\begin{align*}
    U_A(\xi) = Y_T^{e, Z^e,\Gamma^e},
\end{align*}
Define for all $t\in[0,T]$ the processes
\begin{align*}
    & Z_t =: -\frac{Z_t^e}{\gamma Y_t^{e,y_0^e, Z^e,\Gamma^e}}, \; \Gamma_t := -\frac{\Gamma_t^e}{\gamma Y_t^{e, y_0^e, Z^e,\Gamma^e}},\; Y_t^{y_0,Z,\Gamma} = y_0 + \int_0^T Z_s \d B_s + \frac{1}{2}\mathrm{Tr}\big[\big(\Gamma_s + \gamma Z_s Z_s^\top \d\langle B\rangle_s \big] - H\big(s,Z_s,\Gamma_s\big) \d s,
\end{align*}
where $H:[0,T]\times\mathbb{R}^{d^g+d^c+2}\times\mathbb{S}_{d^g+d^c+2}(\mathbb{R})\longrightarrow \mathbb{R}$ is defined by $H(t,z,g) = \sup_{p\in K} h(t,z,g,p)$ and
\begin{align}\label{zg_admissible}
\begin{split}
     \mathcal{ZG}:= \bigg\{(Z_t,\Gamma_t)_{t\in[0,T]}&: \mathbb{R}^{d^g+d^c+2}\times\mathbb{S}_{d^g+d^c+2}(\mathbb{R})\text{-valued, }\mathbb{F}\text{-predictable processes s.t}\\
    & \big(-\gamma Z_t U_A(Y_t^{y_0,Z,\Gamma}),-\gamma \Gamma_t U_A(Y_t^{y_0,Z,\Gamma})\big)_{t\in[0,T]}\in \mathcal{ZG}^e \bigg\}.
\end{split}
\end{align}
An application of It\=o's formula leads to $\xi = Y_T^{y_0,Z,\Gamma}$. Thus, we obtain the desired representation for admissible contracts and $V_A( Y_T^{y_0,Z,\Gamma})=U_A(y_0)$. The characterisation of $\mathcal{A}(Y_T^{y_0,Z,\Gamma})$ is a direct consequence of \citeauthor*{cvitanic2018dynamic} \cite[Proposition 3.3]{cvitanic2018dynamic}. 

\section{Green investments with stochastic interest rates}\label{sec_stoch_rates}

\subsection{Framework}

In the article, we considered a deterministic structure for the short-term rates. However, this omits some important stylised facts of the yield curve. In this section we show that at the expense of the use of stochastic control, the government can provide incentives based on short-term rates following a one factor stochastic model. 

\medskip
We now assume that the vectors of short rate dynamics of the green bonds are given by
\begin{align}\label{eq_stochastic_rates}
    \d r_t^{g} := a^{g}(t,r_t^g) \mathrm{d}t + \text{diag}(b^{g}) \d W_t^{g,r},
\end{align}
where $b^g\in \mathbb{R}^{d^g}_+$, $a^{g}:[0,T]\times \mathbb{R}^{d^g}\longrightarrow \mathbb{R}^{d^g}$ and $W^{g,r}$ is a $d^g$-dimensional Brownian motion of correlation matrix $\Sigma^{g,r}$. 
\begin{remark}
For notational simplicity, we assume no dependence between the risk sources of the short-term rates and the ones of the bonds. Allowing such dependence is straightforward and does not lead to a higher dimension of the control problem. 
\end{remark}
We contract only on the portfolio process, the risk factors of the green bonds and of the stochastic short-term rate of the green bonds, and the risk factor of the index of conventional bonds. The new sets of state variables are 
\begin{align*}
    B^{\text{\rm obs},S} =
    \begin{pmatrix}
    X \\
    W^g\\
    r^g \\
    W^I
    \end{pmatrix},
    \; B^{\text{\sout{\rm obs}},S} = W^c,
\end{align*}
where the superscript $S$ stands for stochastic, which can be written as
\begin{align*}
   & \d B_t^{\text{\rm obs},S} := \mu^{\text{\rm obs},S}(t,\pi_t,r_t^g) \mathrm{d}t + \Sigma^{\text{\rm obs},S}(t,\pi_t) \d W_t, \; \d B_t^{\text{\sout{\rm obs}},S} := \mu^{\text{\sout{\rm obs}},S}(t) \mathrm{d}t+ \Sigma^{\text{\sout{\rm obs}},S}(t) \d W_t,
\end{align*}
where for all $t\in[0,T],p=(p^g,p^c,p^I)\in \mathbb{R}^{d^g}\times\mathbb{R}^{d^c}\times\mathbb{R},r^g\in \mathbb{R}^{d^g}$
\begin{align*}
    & W_t := 
    \begin{pmatrix}
    W_t^g \\
    W_t^{g,r} \\
    W_t^I \\
    W_t^c \\
    \end{pmatrix}, \; \mu^{\text{\rm obs},S}(t,p,r^g) := 
    \begin{pmatrix}
        p^g \cdot \big(r^g+\eta^g(t)\circ\sigma^g(t)\big) + p^c \cdot \big(r^c(t)+\eta^c(t)\circ\sigma^c(t)\big) + p^I \mu^I(t) \\
        \mathbf{0}_{d^g,1}\\
        a^g(t,r^g) \\ 
        0
    \end{pmatrix}, \\
   & \Sigma^{\text{\rm obs},S}(t,p) := 
    \begin{pmatrix}
        (p^g\circ\sigma(t)^g)^{\top}  & \mathbf{0}_{1,d^g}  & p^I \sigma^I(t) & (p^c\circ\sigma(t)^c)^{\top} \\
        I_{d^g} & \mathbf{0}_{d^g,d^g} & \mathbf{0}_{d^g,1} & \mathbf{0}_{d^g,d^c}    \\
        \mathbf{0}_{d^g,d^g} & \text{diag}(b^g) & \mathbf{0}_{d^g,1} & \mathbf{0}_{d^g,d^c}   \\
        \mathbf{0}_{1,d^g} & \mathbf{0}_{1,d^g} & 1 & \mathbf{0}_{1,d^c} 
    \end{pmatrix}, \\
    & \mu^{\text{\sout{\rm obs}},S}(t) = 
    \begin{pmatrix}
    \mathbf{0}_{d^c,1}
    \end{pmatrix}, \; \Sigma^{\text{\sout{\rm obs}},S}(t) = 
    \begin{pmatrix}
        \mathbf{0}_{d^c,d^g} & I_{d^c}  & \mathbf{0}_{d^c,1}& \mathbf{0}_{d^c,d^g} \\
    \end{pmatrix}. 
\end{align*}
We now specify the new set of admissible contracts that we consider for the incentives proposed by the government. 

\subsection{Representation of admissible contracts}

Define $\mathcal{C}^S$ as the set of admissible contracts in the case of stochastic short-term rates (the admissibility conditions are the same as for the set $\mathcal{C}$) and for any $\pi\in\mathcal{A}$ we introduce the following quantities
\begin{align*}
    B^S:=
    \begin{pmatrix}
    B^{\text{\rm obs},S} \\
    B^{\text{\sout{\rm obs}},S}
    \end{pmatrix}, \; \mu^S(t,\pi) := \begin{pmatrix}
    \mu^{\text{\rm obs},S}(t,\pi) \\
    \mu^{\text{\sout{\rm obs}},S}
    \end{pmatrix}, \; 
     \Sigma^S(t,\pi) := \begin{pmatrix}
    \Sigma^{\text{\rm obs},S}(t,\pi) \\
    \Sigma^{\text{\sout{\rm obs}},S}
    \end{pmatrix}.
\end{align*}
We define $h^S:[0,T]\times \mathbb{R}^{2d^g+d^c+2}\times \mathbb{S}_{2d^g+d^c+2}(\mathbb{R})\times \mathbb{R}^{d^g}\times K\longrightarrow \mathbb{R}$ such that 
\begin{align*}
    h^S(t,z,g,r^g,p) = -k(p) + z \cdot \mu^S(t,p,r^g) +  \frac{1}{2}\mathrm{Tr}\Big[ g\Sigma^S(t,p)\Sigma (\Sigma^S(t,p))^{\top}  \Big], 
\end{align*}
and for all $(t,z,g,r^g)\in [0,T]\times \mathbb{R}^{2d^g+d^c+2}\times\mathbb{S}_{2d^g+d^c+2}(\mathbb{R})$ we define 
\begin{align*}
    \mathcal{O}^S(t,z,g,r^g):= \Big\{\hat p\in K: \hat p \in \underset{p\in K}{\text{argmax }}h^S(t,z,g,r^g,p)\Big\},
\end{align*}
as the set of maximisers of $h^S$ with respect to its last variable for $(t,z,g,r^g)$ fixed. Following \citeauthor{schal1974selection} \cite{schal1974selection}, there exists at least one Borel-measurable map $\hat\pi:[0,T]\times \mathbb{R}^{2d^g+d^c+2}\times\mathbb{S}_{2d^g+d^c+2}(\mathbb{R})\times\mathbb{R}^{d^g}\longrightarrow K$ such that for every $(t,z,g,r^g)\in [0,T]\times \mathbb{R}^{2d^g+d^c+2}\times\mathbb{S}_{2d^g+d^c+2}(\mathbb{R})\times\mathbb{R}^{d^g}$, $\hat\pi(t,z,g,r^g)\in \mathcal{O}^S(t,z,g,r^g)$. We denote by $\mathcal{O}^S$ the set of all such maps. By analogy with \Cref{thm_1_admissible_contract}, the following theorem states the form of any admissible contracts in this setting.
\begin{theorem}
Without reducing the utility of the Principal, we can restrict the study of admissible contracts to the set $\mathcal{C}^S_1$  where any $\xi\in\mathcal{C}^S_1$ is of the form $\xi=Y_T^{y_0,Z^S,\Gamma^S,\hat\pi}$ where for all $t\in[0,T]$,
\begin{align}\label{opt_contract_sto_1}
    Y_t^{y_0,Z^S,\Gamma^S,\hat\pi} :=  y_0 + \int_0^t Z^S_s\cdot \d B_s + \frac{1}{2}\mathrm{Tr}\Big[\big(\Gamma^S_s + \gamma Z^S_s (Z_s^S)^{\top}\big)\d\langle B^S\rangle_s\Big] - h^S\big(s,Z^S_s,\Gamma^S_s,r_s^g,\hat\pi(s,Z_s,\Gamma_s,r_s^g)\big)\mathrm{d}s,
\end{align}
where $\hat\pi\in\mathcal{O}^S$ and $(Z^S_t)_{t\in[0,T]}, (\Gamma_t^S)_{t\in[0,T]}$ are respectively $\mathbb{R}^{2d^g+2+d^c}$ and $\mathbb{S}_{2d^g+2+d^c}(\mathbb{R})$-valued, $\F$-predictable processes satisfying similar conditions as the elements of $\mathcal{ZG}$. We denote the set of admissible incentives as $\mathcal{ZG}^S$. Moreover in the present case of stochastic rates for green bonds
\begin{align*}
    V^A(Y_T^{y_0,Z^S,\Gamma^S,\hat\pi})=U_A(y_0),\; 
    \mathcal{A}\big(Y_T^{y_0,Z^S,\Gamma^S,\hat \pi}\big)= \Big\{\big(\hat\pi(t,Z^S_t,\Gamma^S_t,r_t^g)\big)_{t\in[0,T]}, \hat{\pi}\in\mathcal{O}^S, (Z^S_t,\Gamma^S_t)_{t\in[0,T]}\in\mathcal{ZG}^S\Big\}.
\end{align*}

\end{theorem}
We now set
\begin{align*}
    Z^S_t = 
    \begin{pmatrix}
    Z_t^{\text{\rm obs},S}\\
    Z_t^{\text{\sout{\rm obs}},S}
    \end{pmatrix}, \;
    \Gamma_t^S = 
    \begin{pmatrix}
    \Gamma_t^{\text{\rm obs},S}& \Gamma_t^{\text{\rm obs},\text{\sout{\rm obs}},S}\\
    \Gamma_t^{\text{\rm obs},\text{\sout{\rm obs}},S} & \Gamma_t^{\text{\sout{\rm obs}},S}
    \end{pmatrix},
\end{align*}
where for all $t\in[0,T]$
\begin{align*}
    & Z_t^{\text{\rm obs},S}\in \mathbb{R}^{2d^g+2}, Z_t^{\text{\sout{\rm obs}},S}\in  \mathbb{R}^{d^c}, \Gamma_t^{\text{\rm obs},S} \in \mathbb{S}_{2d^g+2}(\mathbb{R}), \Gamma_t^{\text{\sout{\rm obs}},S} \in \mathbb{S}_{d^c}(\mathbb{R}), \Gamma_t^{\text{\rm obs},\text{\sout{\rm obs}},S} \in \mathcal{M}_{2d^g+2,d^c}(\mathbb{R}).
\end{align*}
We define $h^{\text{obs},S}:[0,T]\times \mathbb{R}^{2d^g+2}\times \mathbb{S}_{2d^g+2}(\mathbb{R})\times\mathbb{R}^{d^g}\times K\longrightarrow \mathbb{R}$ such that 
\begin{align*}
    h^{\text{obs}}(t,z^{\text{obs},S},g^{\text{obs},S},r^g,p) = -k(p) + z^{\text{\rm obs},S} \cdot \mu^{\text{obs},S}(t,p,r^g) +  \frac{1}{2}\mathrm{Tr}\Big[ g^{\text{\rm obs},S}\Sigma^{\text{\rm obs},S}(t,p)\Sigma (\Sigma^{\text{\rm obs},S}(t,p))^{\top}  \Big], 
\end{align*}
and for all $(t,z^{\text{obs},S},g^{\text{obs},S},r^g)\in [0,T]\times \mathbb{R}^{2d^g+2}\times\mathbb{S}_{2d^g+2}(\mathbb{R})\times\mathbb{R}^{d^g}$ we define
\begin{align*}
    \mathcal{O}^{\text{obs},S}(t,z^{\text{obs}},g^{\text{obs}},r^g) := \Big\{\hat p\in K:  \hat p\in \underset{p\in K}{\text{argmax }}h^{\text{obs},S}(t,z^{\text{obs},S},g^{\text{obs},S},r^g,p)\Big\}.
\end{align*}
Using again \citeauthor{schal1974selection} \cite{schal1974selection}, there exists at least one Borel-measurable map $\hat\pi:[0,T]\times \mathbb{R}^{2d^g+2}\times\mathbb{S}_{2d^g+2}(\mathbb{R})\times\mathbb{R}^{d^g}\longrightarrow~B$ such that for every $(t,z^{\text{obs},S},g^{\text{obs},S},r^g)\in [0,T]\times \mathbb{R}^{2d^g+2}\times\mathbb{S}_{2d^g+2}(\mathbb{R})\times\mathbb{R}^{d^g}$, $\hat\pi(t,z^{\text{obs},S},g^{\text{obs},S},r^g)\in \mathcal{O}^{\text{obs},S}(t,z^{\text{obs},S},g^{\text{obs},S},r^g)$ and $\mathcal{O}^{\text{obs},S}$ denotes the set of all such maps. We consider the subset of admissible contracts
\begin{align*}
\mathcal{C}_2^S:=\Big\{Y_T^{y_0,Z^S,\Gamma^S,\hat\pi}\in \mathcal{C}_1^S:Z^{\text{\sout{\rm obs}},S}=\mathbf{0}_{d^c},\Gamma^{\text{\sout{\rm obs}},S}=\mathbf{0}_{d^c,d^c},\Gamma^{\text{\rm obs},\text{\sout{\rm obs}},S}=\mathbf{0}_{2d^g+2,d^c} \Big\} \subset \mathcal{C}_1^S \subset \mathcal{C}^S,  
\end{align*}
where any contract in $\mathcal{C}_2^S$ is of the form $Y_T^{y_0,Z^{\text{obs},S},\Gamma^{\text{obs},S},\hat\pi}$ where for all $t\in[0,T]$,
\begin{align}\label{opt_contract_sto}
\begin{split}
  Y_t^{y_0,Z^{\text{obs},S},\Gamma^{\text{obs},S},\hat\pi} :=  y_0 + \int_0^t & Z_s^{\text{\rm obs},S}\cdot \d B_s^{\text{\rm obs},S} + \frac{1}{2}\mathrm{Tr}\Big[\big(\Gamma^{\text{\rm obs},S}_s + \gamma Z_s^{\text{\rm obs},S}(Z_s^{\text{\rm obs},S})^{\top}\big)d\langle B^{\text{\rm obs},S}\rangle_s\Big] \\
    & - h^{\text{\rm obs},S}\Big(s,Z_s^{\text{\rm obs}},\Gamma_s^{\text{\rm obs}},r_s^g,\hat\pi\big(s,Z_s^{\text{\rm obs},S},\Gamma_s^{\text{\rm obs},S},r_s^g\big)\Big)\mathrm{d}s,    
\end{split}
\end{align}
where $y_0\geq 0$, $\hat\pi\in\mathcal{O}^{\text{obs},S}$ and $(Z^{\text{obs},S},\Gamma^{\text{obs},S})\in\mathcal{ZG}^{\text{obs},S}$ with 
\begin{align*}
    \mathcal{ZG}^{\text{obs},S}:= \Big\{ &(Z^{\text{obs},S},\Gamma^{\text{obs},S}): \mathbb{R}^{2d^g+2}\times\mathbb{S}_{2d^g+2}(\mathbb{R})\text{-valued, }\mathbb{F}\text{-predictable s.t } Y_T^{y_0,Z^{\text{obs},S},\Gamma^{\text{obs},S},\hat\pi}\in\mathcal{C}_2^S \Big\}.
\end{align*}
We can now formulate the stochastic control problem faced by the government. 

\subsection{The Hamilton-Jacobi-Bellman equation}

Let us define the process $(Q^{y_0,Z^{\text{obs},S},\Gamma^{\text{obs},S},\hat\pi}_t)_{t\in[0,T]}$ where for all $(t,y_0,Z^{\text{obs},S},\Gamma^{\text{obs},S},\hat\pi)\in[0,T]\times\mathbb{R}\times\mathcal{ZG}^{\text{obs},S}\times\mathcal{O}^{\text{obs},S}$
\begin{align*}
    Q_t^{y_0,Z^{\text{obs},S},\Gamma^{\text{obs},S},\hat\pi} & := X_t  - \int_0^t \sum_{i=1}^{d^g}\big(G_i -\hat\pi_i^{g}(s,Z_s^{\text{\rm obs},S},\Gamma_s^{\text{\rm obs},S},r_s^g)\big)^2 ds -Y_t^{y_0,Z^{\text{obs},S},\Gamma^{\text{obs},S},\hat\pi}.
\end{align*}
The optimisation problem of the government that we consider here is 
\begin{align}\label{pb_principal_2_stoc}
\begin{split}
    \widetilde{V}_0^P = \sup_{y_0\geq 0}\sup_{(Z^{\text{obs},S},\Gamma^{\text{obs},S},\hat\pi)\in \mathcal{ZG}^{\text{obs},S}\times\mathcal{O}^{\text{obs},S}}\mathbb{E}^{\hat\pi(Z^{\text{obs},S},\Gamma^{\text{obs},S})}&\Bigg[-\exp\Big(-\nu Q_T^{y_0,Z^{\text{obs},S},\Gamma^{\text{obs},S},\hat\pi}\Big)\Bigg].    
\end{split}
\end{align}
Due to the presence of state variables in the best response of the Agent, the optimal control of the Principal will no longer be deterministic, and we have to rely on the Hamilton-Jacobi-Bellman formulation of the stochastic control problem. First, we note that the supremum over $y_0$ is attained by setting $y_0=0$. Next, the state variables of the control problem are $\big(t,B_t^{\text{obs},S},Q_t^{0,Z^{\text{obs},S},\Gamma^{\text{obs},S},\hat\pi}\big)$ and as it is standard in control problems with CARA utility function, the last variable can be simplified. Define $P^S=\mathbb{R}^{2d^g+2}\times \mathbb{S}_{2d^g+2}(\mathbb{R})$, and the Hamiltonian 
$H^{\hat\pi}:[0,T]\times P^S \times  \mathbb{R}\times P^S \longrightarrow\mathbb{R}$ 
\begin{align*}
    & H^{\hat{\pi}}(t,z,g,u,u_b,u_{bb}) := \nu u  \bigg( z \cdot \mu^{\text{\rm obs},S}\big(t,\hat\pi(t,z,g,r^g),r^g\big) +  \sum_{i=1}^{d^g}\big(G_i -\hat\pi_i^{g}(t,z,g,r^g)\big)^2 \\
    & + \frac{1}{2}\mathrm{Tr}\bigg[(g+\gamma z z^{\top})\Sigma^{\text{\rm obs},S}\big(t,\hat\pi(t,z,g,r^g)\big)\Big(\Sigma^{\text{\rm obs},S}\big(t,\hat\pi(t,z,g,r^g)\big)\Big)^{\top}\bigg] - h^{\text{\rm obs},S}\big(t,z,g,r^g,\hat\pi(t,z,g,r^g)\big)  \bigg) \\
    & + \frac{1}{2} \nu^2 u \mathrm{Tr}\bigg[z z^{\top} \Sigma^{\text{\rm obs},S}\big(t,\hat\pi(t,z,g,r^g)\big)\Big(\Sigma^{\text{\rm obs},S}\big(t,\hat\pi(t,z,g,r^g)\big)\Big)^{\top}\bigg] + u_b \cdot \mu^{\text{\rm obs},S}\big(t,\hat\pi(t,z,g,r^g),r^g\big) \\
    & + \frac{1}{2}\mathrm{Tr}\bigg[ \Sigma^{\text{\rm obs},S}\big(t,\hat\pi(t,z,g,r^g)\big)\Big(\Sigma^{\text{\rm obs},S}\big(t,\hat\pi(t,z,g,r^g)\big)\Big)^{\top} u_{bb} \bigg].
\end{align*}
The value function of the control problem of the Principal is solution of the following Hamilton-Jacobi-Bellman equation 
\begin{align}\label{hjb_stoc_rates}
\begin{split}
\begin{cases}
     \partial_t U(t,b) + \underset{(z,g,\hat\pi)\in P^S \times \mathcal{O}^{\text{obs},S}}{\sup } H^{\hat{\pi}}\big(t,z,g,b,U,U_b,U_{bb}\big) =0,\\
     U(T,b)= -1,
\end{cases}
\end{split}
\end{align}
where $U:[0,T]\times\mathbb{R}^{2d^g+2}\longrightarrow\mathbb{R}$ and for all $(i,j)\in \{1,\dots,2d^g+2\}$, $(U_b)_i = \partial_{b_i} U, (U_{bb})_{i,j}=\partial_{b_i b_j}U$, in the sense that $\widetilde V^P_0=U(0,b_0)$ where $B_0^{\text{obs},S}=b_0$ and $y_0=0$. Thus, the incentives provided to the investor are obtained up to the resolution of a $(2d^g + 2)$-dimensional HJB equation. Although it provides greater flexibility on the modelling of short-term rates, this approach can only be applied to a small portfolio of bonds using classic numerical schemes on sparse grids. 

\end{appendix}
\bibliography{biblio.bib}

\end{document}